\documentclass[12pt]{amsart}
\reversemarginpar
\usepackage{amsmath}
\usepackage{amssymb}
\usepackage{graphicx}
\usepackage{graphicx}
\usepackage{dcolumn}
\usepackage{bm}
\usepackage{amsfonts}
\usepackage{latexsym}
\usepackage{epsfig}
\usepackage{subfigure}

\begin{document}
\newcommand{\commentout}[1]{}

\newcommand{\nwc}{\newcommand}
\newcommand{\bz}{{\mathbf z}}
\newcommand{\sqk}{\sqrt{\ks}}
\newcommand{\sqkone}{\sqrt{|\ks_1|}}
\newcommand{\sqktwo}{\sqrt{|\ks_2|}}
\newcommand{\invsqkone}{|\ks_1|^{-1/2}}
\newcommand{\invsqktwo}{|\ks_2|^{-1/2}}
\newcommand{\partz}{\frac{\partial}{\partial z}}
\newcommand{\grady}{\nabla_{\ba}}
\newcommand{\gradp}{\nabla_{\bp}}
\newcommand{\gradx}{\nabla_{\bx}}
\newcommand{\invf}{\cF^{-1}_2}
\newcommand{\myphi}{\Phi_{(\eta,\rho)}}
\newcommand{\minrg}{|\min{(\rho,\gamma^{-1})}|}
\newcommand{\al}{\alpha}
\newcommand{\xvec}{\vec{\mathbf x}}
\newcommand{\kvec}{{\vec{\mathbf k}}}
\newcommand{\lt}{\left}
\newcommand{\ksq}{\sqrt{\ks}}
\newcommand{\rt}{\right}
\newcommand{\ga}{\gamma}
\newcommand{\vas}{\varepsilon}
\newcommand{\lan}{\left\langle}
\newcommand{\ran}{\right\rangle}
\newcommand{\tvas}{{W_z^\vas}}
\newcommand{\psiep}{{W_z^\vas}}
\newcommand{\wep}{{W^\vas}}
\newcommand{\weptil}{{\tilde{W}^\vas}}
\newcommand{\wepz}{{W_z^\vas}}
\newcommand{\weps}{{W_s^\ep}}
\newcommand{\wepsp}{{W_s^{\ep'}}}
\newcommand{\wepzp}{{W_z^{\vas'}}}
\newcommand{\wepztil}{{\tilde{W}_z^\vas}}
\newcommand{\vvas}{{\tilde{\ml L}_z^\vas}}
\newcommand{\veptil}{{\tilde{\ml L}_z^\vas}}
\newcommand{\cvc}{{{\ml L}^{\ep*}_z}}
\newcommand{\cvcp}{{{\ml L}^{\ep*'}_z}}
\newcommand{\cvp}{{{\ml L}^{\ep*'}_z}}
\newcommand{\cvtil}{{\tilde{\ml L}^{\ep*}_z}}
\newcommand{\cvtilp}{{\tilde{\ml L}^{\ep*'}_z}}
\newcommand{\vtil}{{\tilde{V}^\ep_z}}
\newcommand{\ktil}{\tilde{K}}
\newcommand{\n}{\nabla}
\newcommand{\tkappa}{\tilde\kappa}
\newcommand{\ks}{{\omega}}
\newcommand{\bx}{\mb x}
\newcommand{\br}{\mb r}
\nwc{\bR}{\mb R}
\nwc{\bH}{{\mb H}}
\newcommand{\bu}{\mathbf u}
\nwc{\bxp}{{{\mathbf x}}}
\nwc{\bap}{{{\mathbf y}}}
\newcommand{\bD}{\mathbf D}
\newcommand{\bA}{\mathbf A}
\nwc{\bPhi}{\mathbf{\Phi}}
\nwc{\bh}{\mathbf h}
\newcommand{\bB}{\mathbf B}
\newcommand{\bC}{\mathbf C}
\newcommand{\bp}{\mathbf p}
\newcommand{\bq}{\mathbf q}
\nwc{\bI}{\mathbf I}
\nwc{\bP}{\mathbf P}
\nwc{\bs}{\mathbf s}
\nwc{\bX}{\mathbf X}
\newcommand{\pdg}{\bp\cdot\nabla}
\newcommand{\pdgx}{\bp\cdot\nabla_\bx}
\newcommand{\one}{1\hspace{-4.4pt}1}
\newcommand{\corr}{r_{\eta,\rho}}
\newcommand{\rinf}{r_{\eta,\infty}}
\newcommand{\rzero}{r_{0,\rho}}
\newcommand{\rzeroinf}{r_{0,\infty}}
\nwc{\om}{\omega}

\nwc{\nwt}{\newtheorem}
\nwc{\xp}{{x^{\perp}}}
\nwc{\yp}{{y^{\perp}}}
\nwt{remark}{Remark}
\nwt{definition}{Definition} 

\nwc{\ba}{{\mb a}}
\nwc{\bal}{\begin{align}}
\nwc{\be}{\begin{equation}}
\nwc{\ben}{\begin{equation*}}
\nwc{\bea}{\begin{eqnarray}}
\nwc{\beq}{\begin{eqnarray}}
\nwc{\bean}{\begin{eqnarray*}}
\nwc{\beqn}{\begin{eqnarray*}}
\nwc{\beqast}{\begin{eqnarray*}}

\nwc{\eal}{\end{align}}
\nwc{\ee}{\end{equation}}
\nwc{\een}{\end{equation*}}
\nwc{\eea}{\end{eqnarray}}
\nwc{\eeq}{\end{eqnarray}}
\nwc{\eean}{\end{eqnarray*}}
\nwc{\eeqn}{\end{eqnarray*}}
\nwc{\eeqast}{\end{eqnarray*}}

\nwc{\vep}{\varepsilon}
\nwc{\ep}{\epsilon}
\nwc{\ept}{\epsilon}
\nwc{\vrho}{\varrho}
\nwc{\orho}{\bar\varrho}
\nwc{\ou}{\bar u}
\nwc{\vpsi}{\varpsi}
\nwc{\lamb}{\lambda}
\nwc{\Var}{{\rm Var}}

\nwt{proposition}{Proposition}
\nwt{theorem}{Theorem}
\nwt{summary}{Summary}
\nwt{lemma}{Lemma}
\nwt{cor}{Corollary}
\nwc{\nn}{\nonumber}
\nwc{\mf}{\mathbf}
\nwc{\mb}{\mathbf}
\nwc{\ml}{\mathcal}

\nwc{\IA}{\mathbb{A}} 
\nwc{\bi}{\mathbf i}
\nwc{\bo}{\mathbf o}
\nwc{\IB}{\mathbb{B}}
\nwc{\IC}{\mathbb{C}} 
\nwc{\ID}{\mathbb{D}} 
\nwc{\IM}{\mathbb{M}} 
\nwc{\IP}{\mathbb{P}} 
\nwc{\II}{\mathbb{I}} 
\nwc{\IE}{\mathbb{E}} 
\nwc{\IF}{\mathbb{F}} 
\nwc{\IG}{\mathbb{G}} 
\nwc{\IN}{\mathbb{N}} 
\nwc{\IQ}{\mathbb{Q}} 
\nwc{\IR}{\mathbb{R}} 
\nwc{\IT}{\mathbb{T}} 
\nwc{\IZ}{\mathbb{Z}} 

\nwc{\cE}{{\ml E}}
\nwc{\cP}{{\ml P}}
\nwc{\cQ}{{\ml Q}}
\nwc{\cL}{{\ml L}}
\nwc{\cX}{{\ml X}}
\nwc{\cW}{{\ml W}}
\nwc{\cZ}{{\ml Z}}
\nwc{\cR}{{\ml R}}
\nwc{\cV}{{\ml L}}
\nwc{\cT}{{\ml T}}
\nwc{\crV}{{\ml L}_{(\delta,\rho)}}
\nwc{\cC}{{\ml C}}
\nwc{\cO}{{\ml O}}
\nwc{\cA}{{\ml A}}
\nwc{\cK}{{\ml K}}
\nwc{\cB}{{\ml B}}
\nwc{\cD}{{\ml D}}
\nwc{\cF}{{\ml F}}
\nwc{\cS}{{\ml S}}
\nwc{\cM}{{\ml M}}
\nwc{\cG}{{\ml G}}
\nwc{\cH}{{\ml H}}
\nwc{\bk}{{\mb k}}
\nwc{\cbz}{\overline{\cB}_z}
\nwc{\supp}{{\hbox{supp}(\theta)}}
\nwc{\fR}{\mathfrak{R}}
\nwc{\bY}{\mathbf Y}
\newcommand{\mbr}{\mb r}
\nwc{\pft}{\cF^{-1}_2}
\nwc{\bU}{{\mb U}}
\nwc{\bG}{{\mb G}}

\title{Compressed Remote Sensing of Sparse  Objects 
}

\author{Albert C.  Fannjiang}
  \address{
   Department of Mathematics,
    University of California, Davis, CA 95616-8633}
\email{
fannjiang@math.ucdavis.edu}
  \thanks{
  The research supported in part by  DARPA Grant N00014-02-1-0603 and NSF Grant DMS 0811169}
   \author{ Pengchong Yan}
    \address{Applied and Computational Mathematics, California Institute of Technology, CA 91125}  \email{yan@acm.caltech.edu}

   \author{ Thomas Strohmer}
      \address{
   Department of Mathematics,
    University of California, Davis, CA 95616-8633}
   \email{strohmer@math.ucdavis.edu}
 
       \begin{abstract}
 The linear inverse source and scattering problems are  studied  
from the perspective of  compressed sensing, in particular
the idea that sufficient incoherence and sparsity guarantee
uniqueness of the solution.  By introducing the
sensor as well as target ensembles,  the 
maximum  number of recoverable targets (MNRT) is proved
to be at least proportional to the number of measurement
data modulo a log-square factor with overwhelming
probability. 

Important contributions include  the discoveries
of the threshold aperture, consistent with the classical Rayleigh criterion,  and the 
decoherence effect  induced  by random antenna locations.  

The prediction of theorems are confirmed by numerical simulations.

       \end{abstract}
       
       \maketitle

 \section{Introduction}
We consider the imaging problem in  the form of  inverse  source or scattering problem which  has
 wide-range applications such as  radar, sonar and computed tomography. The imaging problem  is typically
 plagued by nonuniqueness and instability and
 hence mathematically  challenging.
 Traditional  methods such as  matched field 
processing ~\cite{Tol} are  limited in 
the number of targets  that  can be reliably recovered at high resolution.
They often fail to detect a substantial number of targets, while at the
same time they tend to produce artifacts obscuring the real
target images. These limitations are due to the presence of noise and the fact that  the imaging problem is in practice underdetermined. The standard regularization
methods can handle to some extent the problem with noise
but are inadequate to remedy the issue of nonuniqueness of the solution.

In this paper we utilize the fact that in many imaging 
applications the targets are sparse in the sense that they typically occupy a 
small fraction  of  the overall region of interest (the target domain).
This sparsity assumption suggests to approach the imaging problem by
using the framework of {\em compressed sensing}. 

At the core of compressed 
sensing lies the following problem (here we focus,
as is common in the compressed sensing community, on the discrete setting).
Assume $X \in \IC^m$ is a signal that is sparse, i.e., the number of its
non-zero components (measured by the $\ell_0$-quasinorm $\|X\|_0$ which
is simply the number of non-zero entries of $X$) satisfies $s:=\|X\|_0 \ll
m$. Let $Y\in \IC^n$ be the measurement data vector.  
We explore in this paper the {\em linear} inverse problem
which can be formulated as 
$Y=\bA X$  where $\bA$ is an $n\times m$ matrix with $
n\ll m$. 
The goal is to recover $X$, 
given the data vector $Y$ and the sensing matrix  $\bA$ of full rank. As $n\ll m$,  $\bA X=Y$ is severely  underdetermined and  unique reconstruction of $X$ is 
in general impossible.

However, due to the sparsity of $X$ one can
compute $X$ by solving the optimization problem
\begin{equation}
\min \|X\|_0 \qquad \text{s.t.} \,\, \bA X=Y.  \tag{L0}
\label{L0}
\end{equation}
Since~\eqref{L0} is NP-hard and thus computationally infeasible, we 
consider instead its convex relaxation, also known as Basis Pursuit (BP),
\begin{equation}
\min \|X\|_1 \qquad \text{s.t.} \,\, \bA X=Y  \tag{L1}
\label{L1}
\end{equation}
which can be solved by linear and quadratic  programming techniques. 
The amazing discovery due to David Donoho was that under certain conditions
on the matrix $\bA$ and the sparsity of $X$, both
\eqref{L1} and \eqref{L0} have the same  unique solution~\cite{DH}. One such condition is the
{\em Restricted Isometry Property} (RIP)  due to Candes and Tao~\cite{CT}, which 
requires essentially that any $n \times s$ submatrix of
$\bA$ is an approximate isometry. This property is satisfied
by a number of matrices such as Gaussian random matrices or random partial
Fourier matrices~\cite{CT, CRT1, RV}. In that case, as long as
$s\leq {\mathcal O}(n/\log(m))$, with high probability the solution
of~\eqref{L1} will indeed coincide with the solution of~\eqref{L0}.
Another conditon for which equivalence between (L0) and (L1) can
be proven is based on the {\em incoherence} of the columns of $\bA$,
which refers to the property that the inner product of any two columns 
of $\bA$ is small \cite{DE,GN,Tro}. Moreover,
the performance of BP is stable w.r.t. the presence
of  noise and error \cite{CRT2, DET, Tro3}.  
Finally the computational complexity of BP  can be significantly 
reduced 
by using the various  greedy
algorithms in place of the linear programming technique \cite{DM, NTV, NV,  Tro, Tro3}. The most basic greedy algorithm
relevant here is Orthogonal Matching Pursuit (OMP) which
has been thoroughly analyzed in \cite{Tro}.

For the imaging problem, the sensing  matrix $\bA$ 
represents a physical process (typically  wave propagation) and
thus its entries cannot be arbitrarily chosen at our convenience. Therefore we
cannot simply assume that $\bA$ satisfies any of the conditions
that make compressed sensing work. The few physical parameters that we have control over are
the wavelength $\lambda$ of the probe wave, the
locations and number $n$ of sensors  and the aperture $A$  of the probe array.
This is one of
the reasons that make the practical realization  of compressed sensing
a challenging task.

The  paper is organized as follows. In Section \ref{sec2} we describe the physical setup, formulate the imaging problem
in the framework of compressed sensing and make
qualitative statements of our main results. In Section \ref{sec3}
we prove the main result for  the inverse source problem,
in particular the coherence estimate (Section \ref{sec3.1})
and the spectral norm bound (Section \ref{sec3.2}). 
In Section \ref{sec:ibs}, we prove the main result for the inverse
Born scattering problem for the response matrix  imaging (Section \ref{sec:rm})
and the synthetic aperture  imaging (Section \ref{sec:sa}). 
In Section \ref{sec:num} and Appendix B, we discuss the numerical method
and present simulation results that confirm qualitatively 
the predictions of our theorems.  In Appendix \ref{app:rip} we
discuss the RIP approach to our problems.

\section{Problem formulations  and main results}\label{sec2}
In this paper, we study the inverse source and scattering problems both in the linear regime to suit  the current 
framework of compressed
sensing. 
For simplicity and definiteness  we consider the three dimensional space and assume that all targets are in the transverse plane $\{z=z_0\}$ and all sensors are in another transverse plane $\{z=0\}$. The exact Green function for the Helmholtz equation which governs the monochromatic wave propagation is
\beq
\label{green}
G(\br,\ba)={e^{i\om|\br-\ba|}\over 4\pi |\br-\ba|},\quad
\br=( x,y, z_0),\quad \ba=(\xi,\eta, 0). 
\eeq
We assume that
the phase speed $c=1$ so that the frequency $\omega $
equals 
the wavenumber.

 We  consider the Fresnel diffraction regime where
the distance $z_0$ between the targets and the sensors
is much larger than the wavelength of the probe wave
and  the linear dimensions of the domains \cite{BW} 
\beq
\label{6-2}
z_0\gg A+L ,\quad z_0\gg \lambda
\eeq
where  $L$ is  the linear dimension of the target  domain. 
This is
the remote sensing regime. 

Under (\ref{6-2}) 
the Green function  (\ref{green}) can be approximated by
the universal parabolic  form  \cite{BW}
\beq
\label{8-3}
G(\br,\ba)={e^{i\om z_0}\over 4\pi z_0}  e^{i\om|x-\xi|^2/(2z_0)}e^{i\om|y-\eta|^2/(2z_0)},
\eeq
which is called the paraxial Green function.
This follows from the truncated Taylor expansion of the function $|\br-\ba|$
\[
|\br-\ba|\approx z_0+{|x-\xi|^2\over 2z_0} + {|y-\eta|^2\over 2z_0}
\]
under (\ref{6-2}). 

In the case
of the inverse source problem, the corresponding sensing matrix $\bA$ is essentially made of the paraxial Green function
for various points in the sensor array and  the target domain.
In this set-up, the entries (\ref{8-3})  of the paraxial sensing matrix   have
the same magnitude so  without loss of generality
the  column vectors of $\bA$ are assumed to have unit $\ell^2$-norm.

A key idea  in our construction of a  suitable
sensing matrix is to randomize the locations $\ba_j=(0,\xi_j,\eta_j), j=1,...,n$ of
the $n$ sensors within a fixed aperture (a square of size $A$ for example). Indeed, we assume $\xi_j,\eta_j$ are independent
uniformly distributed in $[0,A]$.
We assume that the antenna elements are  independently
uniformly distributed in a square array $[0,A]\times [0,A]$ in
the plane $\{z=0\}$. 
 Define the {\em sensor ensemble} to be
 the sample space of $n$ i.i.d. uniformly distributed
 points in $[0,A]^2$.
  
  We consider the idealized situation where the locations of the targets are a subset of
a square lattice. More precisely,
let $\cM$ be 
a regular square sub-lattice 
$
\cM=\lt\{\br_i: i=1,...,m\rt\}
$
of mesh size $\ell$ in  the transverse plane $\{z=z_0\}$.
Hence the total number of grid points
 $m$ is  a perfect square. We defer the discussion on
 extended targets to the concluding  section. 

Let $
\cS=\lt\{\br_{j_l}: l=1,...,s\rt\}$ be the set of target locations
and  $\sigma_{j_l}, l=1,...,s$ be the (source or scattering) amplitudes of the targets.  Set $\sigma_i=0,
i\not \in \{j_1,..., j_s\}$. 
Define the target vector $X$ to be $X=(\sigma_j)\in \IC^m$.
  We consider the {\em target ensemble} consisting of target vectors  with at most 
 $s$ non-zero entries  whose
 phases are 
 independently, uniformly distributed in $[0,2\pi]$  and
 whose support  indices  are independently and  randomly selected from the index set $\{1,2,...,m\}$. The number $s=\|X\|_0$ is called
 the sparsity of the target vector.

For source  inversion  the targets emit the paraxial waves described
by (\ref{8-3}) which are then  recorded by the sensors. The measurement vector $Y$ can be written as
\beq
\label{u1}
Y=\bA X
\eeq
where 
the matrix  $\bA =[A_{ij}] \in \IC^{n\times m}$ have  the entries 
 \beq
 \label{1.10}
 A_{ij}=G(\ba_i, \br_j),\quad \forall i=1,...,n,\quad
 j=1,...,m.
 \eeq
 
The first main result proved in this paper can be stated roughly as
follows  (see Theorem \ref{thm-is} and Remark \ref{rmk2}
for the precise statement). 

\medskip

\noindent {\bf Result A}. {\em Suppose
\beq
\label{0.1}
{\ell A\over \lambda z_0}\equiv {1\over \rho} \in \IN. 
\eeq
For the product ensemble of targets and sensors, sources of sparsity up to $\cO(n/(\ln{m})^2) $
can be exactly recovered by BP with overwhelming probability.

When only the sensor ensemble is considered, {\bf all} sources of sparsity up to $\cO(\sqrt{n}) $
can be exactly recovered by BP and OMP with overwhelming probability.

}

\medskip

The relation (\ref{0.1})  indicates the existence  of the threshold, optimal aperture given by $ \lambda z_0/ \ell$ corresponding to $\rho=1$ (see Remark \ref{rmk1}  for more discussion
on this point). 
Since the meshsize $\ell$ has the meaning of resolution,
$\rho=1$ is consistent with  the classical
Rayleigh criterion \cite{BW}
\beq
\label{Rayleigh}
\ell\geq {\lambda z_0\over A}.  
\eeq
Our numerical simulations (Figure \ref{fig1}) indeed  indicate 
 that (\ref{Rayleigh}) is sufficient to realize the
 performance stated in Result A.

Next we consider  two imaging settings where the targets
are scatterers instead of sources.    
For  point scatterers of amplitudes $\sigma_{j_l}$ located
at $\br_{j_l}, l=1,2,3,...s$, the resulting Green function  $\tilde G$, including the multiple scattering effect, obeys  the Lippmann-Schwinger equation
\beqn
\tilde G(\br,\ba_i)=G(\br,\ba_i)
+\sum_{l=1}^s \sigma_{j_l}
 G(\br,\br_{j_l})\tilde  G(\br_{j_l},\ba_i),\quad i=1,...,n.
\eeqn 
The exciting field $\tilde G(\br_{j_l},\ba_i)$ is part
of the unknown and 
can be solved for  from  the so called Foldy-Lax equation (see e.g. \cite{mfirm-ip}
for details). 

Hence, the inverse scattering problem
is intrinsically nonlinear. However, often 
linear scattering model  is a good approximation
and  widely used in, e.g.  radar imaging
   in the regimes of 
   physical optics and geometric optics \cite{Bor, CM}
   (see \cite{Hul}
   for a precise
   formulation of the condition).

   One such model is
   the Born approximation (also known as
   Rayleigh-Gans scattering in optics) in which
   the unknown exciting field is replaced by the incident field
   resulting in
   \beq
   \label{born}
\tilde G(\br,\ba_i)-G(\br,\ba_i)
=\sum_{l=1}^s \sigma_{j_l}
 G(\br,\br_{j_l})G(\br_{j_l},\ba_i),\quad i=1,...,n.
\eeq
The left hand side of (\ref{born}) is precisely the scattered field
when the incident field is emitted from  a point source at $\ba_i$. The Born approximation linearizes the relation between  the scatterers  and the scattered field. The goal of inverse scattering is   to reconstruct the targets given the measurements
of the scattered field. 

For the {\em response matrix} (RM)  imaging
\cite{mfirm-ip, SA-rice}, we use the real array aperture as in the
inverse source problem discussed above except the array is also
the source of $n$ probe waves.   One by one, each antenna of the array emits  an impulse  and the entire array
receives the echo. Each transmitter-receiver pair gives rise to
a datum and there are altogether $n^2$ data forming a
datum  matrix called the response matrix.  These data represent
the responses of the targets to the interrogating waves.  

From (\ref{born}) we
see that the corresponding sensing matrix
$\bA^{\rm RM}$  has the entries
\[
A^{\rm RM}_{lj}=G(\ba_i,\br_j)G(\br_j, \ba_k), \quad
l=1,..,n^2,\quad j=1,...,m
\]
where $l$ is related to $i,k$ as
\[
l=i(n-1)+k.
\]

In the second setting, called the synthetic aperture (SA) imaging, 
the real, physical array consists of only one antenna.
The imaging aperture is synthesized by the antenna taking different transmit-receive positions $\ba_i, i=1,...,n$  \cite{SA-rice}. 

The SA imaging considered here  is motivated by 
  synthetic aperture  radar (SAR) imaging.  
  SAR  is a technique  where a substantial aperture can be
  synthesized  by moving a transmit-receive antenna
  along a trajectory  and repeatedly interrogating a search area
  by firing repeated pulses from the antenna and measuring the responses. This can greatly
  leverage a limited probe resource and has
  many applications in remote sensing. 
  The image formation is typically obtained via the matched filter technique  and
   analyzed in the Born approximation \cite{CM}.
   
 Here   we  consider a simplified set-up, neglecting the Doppler effect
    associated with the relative motion between the antenna and targets. 
In this case, the sensing matrix $\bA^{\rm SA}$ has
the entries
\beq
\label{50}
A^{\rm SA}_{ij}=G^2(\ba_i,\br_j),\quad i=1,...,n,\quad
j=1,...,m.
\eeq
In other words, $A^{\rm SA}_{ij}=A^{\rm RM}_{lj}$
with $l=i(n-1)+i.$
A crucial observation about SA imaging is 
 that 
\beq
\label{51}
G^2(\ba_i,\br_j;\om)\sim G(\ba_i,\br_j;2\om)
\eeq
modulo  a $z_0$-dependent factor which does not
matter.

The following is a rough statement for  inverse Born scattering (Theorems \ref{thm:rm}, \ref{thm-sa} and
Remarks \ref{rmk3}, \ref{rmk4})  proved in Section \ref{sec:ibs}.

\medskip

\noindent {\bf Result B.} {\bf (i)} {\em  
For RM imaging, assume the aperture condition (\ref{0.1}). 

For the product ensemble of sensor and target,  scatterers of sparsity up
to $\cO(n^2/(\ln{m})^2)$ can be reconstructed
exactly by BP with overwhelming probability.

When only  the sensor ensemble is considered, {\bf all }scatterers of sparsity up to $\cO(n) $
can be exactly recovered by BP and OMP with overwhelming probability.

{\bf (ii)} For SA imaging, assume the aperture condition 
\beq
\label{0.2}
2/\rho\in \IN
\eeq 
which is weaker than  (\ref{0.1}). 

For the product ensemble of sensor and target,  scatterers of sparsity up
to $\cO(n/(\ln{m})^2)$ can be reconstructed
exactly by BP with overwhelming probability.

When only the sensor ensemble is considered, {\bf all} scatterers  of sparsity up to $\cO(\sqrt{n}) $
can be exactly recovered by BP and OMP with overwhelming probability.}

\medskip

As a result of the SA aperture condition (\ref{0.2}), the corresponding optimal  aperture is  half of that 
for the inverse source and RM imaging. In other words,
SA can produce the qualitatively optimal performance
with half of the aperture. This two-fold enhancement
of resolving power in SA imaging has been  previously
established for the matched-field imaging technique \cite{SA-rice}. 

Our numerical simulations (Section \ref{sec:num})  confirm qualitatively 
the predictions of Result A and B, in particular the threshold
aperture and the asymptotic number of recoverable targets. 

Currently there are two avenues  to compressed sensing \cite{BDE}:
the incoherence approach and the RIP (restricted isometry property)  approach.
When the RIP approach works, the results are typically 
superior  in that {\em all}  targets under a slightly weaker
sparsity constraint can be uniquely determined by BP
without introducing the target ensemble. We demonstrate the strength of the RIP approach for our problems 
in Appendix \ref{app:rip} (see Theorem \ref{thm:is2}  and Theorem \ref{thm:sa2} for stronger results  than {\bf Result A} and {\bf Result B (ii)}, respectively). However, {\bf Result B(i)}
seems unattainable by the RIP approach at present, particularly
the quadratic-in-$n$ behavior of the sparsity constraint.
On the other hand, the incoherence approach gives a unified treatment
to all three results and therefore is adopted in the main text
of the paper.

\section{Source inversion}\label{sec3}
 
\commentout{
\[
X_i(\omega)=\lt\{
\begin{array}
\sigma_{j_l}(\om)&&, i=j_l\\
0&&, \mbox{else}
\end{array}\rt.
\]
}

Let $G(\br,\ba)$ be the Green function
of the time-invariant medium 
and let $\bG$ be
the Green  vector
\beq
\label{1.0}
\bG(\br) =[G(\br,\ba_1),
G(\br,\ba_2),...,G(\br,\ba_n)]^t
\eeq
where $t$ denotes transpose. For the matrix (\ref{1.10}) define the {\em  coherence}  of the matrix $\bA$ by
\[
\mu(\bA)=\max_{i\neq j} {\lt|\bG^*(\bp_i)
\bG(\bp_j)\rt| \over \|\bG(\bp_i)\|\|\bG(\bp_j)\|}. 
\]

\commentout{
The main thrust of the recent compressed sensing paradigm is that
under certain decoherence condition on $\bA$ and
sparsity condition on  $X$  finding 
the unique $\ell^0$-sparsest solution to (\ref{u1})  can be
reduced to 
the $\ell^1$ minimization problem
\beq
\label{L1}
\inf_{X}\lt\{ \|X\|_1: Y=\bA X\rt\}
\eeq
which goes by the name of basis pursuit (BP) \cite{ CRT1, CRT2, Don1, DE, GN}. Unlike
the $\ell^0$ minimization,  basis pursuit
 is a convex optimization problem and can be solved
by the technique of  linear programming \cite{CRT1, CRT2, Don1}.

The equivalence
between the $\ell^0$ and $\ell^1$ minimization 
is established under the incoherence condition in \cite{GN, DE} and  the restricted isometry condition  in \cite{CRT1, CRT2}. 
}

The following theorem is a reformulation of results due to Tropp \cite{Tropp} and the foundation of the imaging techniques
developed in this paper.
 \begin{theorem}
 \label{tropp}
 Let $X$ be drawn from the target ensemble. 
 Assume that  
 \beq
 \label{M}
 \mu^2 s\leq \lt( 8\ln{{m\over \epsilon}}\rt)^{-1},\quad
 \ep\in (0,1) 
 \eeq
 and that for $p\geq 1$
 \beq
 \label{Op}
 3\lt({p\ln{s}\over 2\ln{{m\over \epsilon}}}\rt)^{1/2}+{s\over m}\|\bA\|_2^2\leq {1\over 4 e^{1/4}}.
 \eeq
 Then  $X$ is the unique solution of BP  with probability $1-2\epsilon-s^{-p}$. 
 Here $\|\bA\|_2$ denotes the spectral norm of $\bA$. 
 \end{theorem}
We explain the connection of the theorem  with \cite{Tropp} in Appendix \ref{sec:pf}.

\begin{theorem}
\label{thm-is} Let the target vector be randomly drawn from
the target ensemble and the antenna array be randomly
drawn from the sensor ensemble and suppose
\beq
\label{aperture}
{\ell A\over \lambda z_0}\equiv {1\over \rho} \in \IN.
\eeq
If 
\beq
\label{m-1}
m\leq {\delta\over 2} e^{K^2/2},\quad \delta, K>0.
\eeq
\commentout{
and
 \beq
 \label{Op2}
 3p\lt({\ln{s}\over 2\ln{{m\over \epsilon}}}\rt)^{1/2}+{2s\over n}\leq {1\over 4 e^{1/4}},\quad p>1
 \eeq
 }
then the targets of sparsity up to
\beq
\label{spark}
s< {n\over  64  \ln{2m\over \delta}\ln{m\over \ep}}
\eeq
 can be recovered exactly by BP with
probability greater than or equal to 
\beq
\label{prob}
\lt[1-2\delta -{\rho n (n-1)^{3/2}\over m^{1/2}}\rt]
\times \lt[1-2\ep-s^{-p}\rt], \quad
 p={\ln{m} -\ln{\ep} \over  288 \sqrt{e}\ln s}.
\eeq
\end{theorem}
\begin{proof}
The proof of the theorem hinges on the following
two estimates.

\begin{theorem}
\label{lemma1}
\label{thm3}
Assume (\ref{aperture}) and
\beq
\label{mesh}
m\leq {\delta\over 2} e^{K^2/2} 
\eeq
for some positive $\delta$ and $K$. 
Then the coherence  of $\bA$ satisfies
\beq
\label{100}
\mu(\bA)\leq \sqrt{2}K/\sqrt{n}
\eeq
with probability greater than  $(1-\delta)^2$. 
\end{theorem} 

\begin{remark}

The general lower bound for coherence
\cite{DGS, Wel}
\[
\sqrt{m-n\over n(m-1)}\leq \mu\leq 1
\]
implies that the coherence bound  (\ref{100}) is optimal
modulo a constant factor.
\end{remark}

\begin{remark}\label{rmk1}
\commentout{
A main contribution of Theorem \ref{thm3}  is the 
discovery of the threshold, optimal aperture given by $\rho=1$ in (\ref{aperture})
which can be interpreted as follows. The classical
Rayleigh criterion says that the two-point resolution
$\ell$ is given by
\[
\ell={\lambda z_0\over A} 
\]
corresponding  to $\rho=1$. 
}
Since the coherence of the sensing matrix should decrease
as  the aperture increases  and since the analysis in Section \ref{sec3.1} shows that the coherence is of the same
order of magnitude as  $n^{-1/2}$ whenever (\ref{aperture}) holds,
simple interpolation leads to the conclusion  that the coherence
should be roughly constant for 
\beq
\label{12-2}
A\geq {\lambda z_0\over \ell}
\eeq
corresponding to $\rho\leq 1$. 
The right hand side of (\ref{12-2}), corresponding to $\rho=1$,  defines the optimal aperture.
\end{remark}

\begin{theorem}
\label{lemma2}
\label{cor1}
The matrix $\bA$ has full rank and 
its spectral norm satisfies the bound
\beq
\label{norm}
\| \bA\|_2^2\leq 2m/n
\eeq
 with probability greater than  
\beq
\label{11-2}
1-{\rho n(n-1)^{3/2}\over m^{1/2}},\quad \rho={\lambda z_0\over \ell A}.
\eeq
\end{theorem}
\begin{remark}\label{rmk2}
By the theorems of Donoho, Elad \cite{DE} and Tropp \cite{Tro},
the targets of sparsity
\[
s<{1\over 2} (1+{1\over \mu(\bA)})
\]
can be recovered exactly by BP as well as
by Orthogonal Matching Pursuit (OMP).

Theorems \ref{lemma1} and \ref{lemma2} imply that
with probability  greater than
\[
1-2\delta- {\rho n(n-1)^{3/2}\over m^{1/2}}
\]
of the sensor ensemble, {\bf all} targets of sparsity
\[
s<{1\over 2} (1+{\sqrt{n}\over \sqrt{2} K})
\]
can be recovered exactly by BP as well as OMP.

\end{remark}
Condition (\ref{spark}) implies the existence of $K$ such that
\beq
\label{K}
2\ln {2m \over \delta}< K^2< {n\over 32 s\ln{m\over \ep}}.
\eeq
As a consequence  (\ref{mesh}) and (\ref{M}) are satisfied with probability
greater than $1-2\delta$ by Theorem \ref{lemma1}.

Now  the norm bound (\ref{norm}) 
implies (\ref{Op})  if 
\beq
 \label{Op2}
 3\lt({p\ln{s}\over 2\ln{{m\over \epsilon}}}\rt)^{1/2}+{2s\over n}\leq {1\over 4 e^{1/4}},\quad p>1,
 \eeq
which in turn follows from  (\ref{spark}) and the condition
 \[
 \ln{2m\over \delta} \ln{m\over \ep}
 \geq {1\over 96} \lt({1\over 2e^{1/4}}- \lt(p\ln s\over 2\ln {m\over \ep}\rt)^{1/2}\rt)^{-1}.
 \]
 Hence for $m\gg s$ (and hence $n\gg s$) we can choose $p$ in (\ref{Op})  to be 
 \[ p={\ln{m} -\ln{ \ep} \over  72\sqrt{e}\ln s}. 
  \]

Since Theorems \ref{lemma1} and \ref{cor1} hold
with probability
greater than 
\[
1-2\delta- {\rho n(n-1)^{3/2}\over m^{1/2}}.
\] 
and since the target ensemble is independent
of the sensor ensemble we have the
 bound (\ref{prob}) for the probability of
exact  recovery.

\end{proof}
\subsection{Proof of Theorem \ref{lemma1}: coherence estimate}\label{sec3.1}
\begin{proof}
Summing over $\ba_l, l=1,...,n$ we obtain 
\beq
\label{1.20}
\sum_{l=1}^n A^*_{li}A_{lj}
&=&e^{i\om(x_{j}^2+y_{j}^2-x_i^2-y_i^2)/(2z_0)}{1\over n}  \sum_{l=1}^n e^{i\xi_l\om(x_i-x_j)/z_0}
e^{i\eta_l\om(y_i-y_j)/z_0}.
\eeq
Define the random variables $X_l, Y_l, l=1,...,n$, as 
\beq
\label{9.1}
X_l&=&\cos{\lt[(\xi_l(x_i-x_j)+\eta_l(y_i-y_j))\om/z_0\rt]}\\
Y_l&=&\sin{\lt[(\xi_l(x_i-x_j)+\eta_l(y_i-y_j))\om/z_0\rt]}\label{9.2}
\eeq
and their respective sums
\beqn
\label{9}
S_n=\sum_{l=1}^n X_l,\quad T_n=\sum_{l=1}^n Y_l.
\eeqn
Then the absolute value of the
right hand side of (\ref{1.20}) is bounded by
\beq
\label{sum}
{1\over n} \lt|S_n+iT_n\rt|\leq {1\over n}\lt(\lt|S_n-\IE S_n\rt|
+\lt|T_n-\IE T_n\rt|+\lt|\IE (S_n+iT_n)\rt|
\rt).
\eeq

To estimate the right hand side of  (\ref{sum}), we recall the Hoeffding inequality
\cite{Hoe}.

\begin{proposition}
Let $X_1, ..., X_n$ be independent random variables. Assume
that $X_l \in [a_l, b_l], l=1,...,n$ almost surely.
Then we have 
\beq
\label{hoeff}
\IP\lt[\lt|S_n -\IE S_n\rt|\geq nt\rt]
\leq 2\exp{\lt[-{2n^2 t^2\over \sum_{l=1}^n (b_l-a_l)^2}\rt]}
\eeq
for all positive values of $t$. 
\end{proposition}
We apply the Hoeffding inequality to both $S_n$ and $T_n$. To
this end, we have $a_l=-1, b_l=1, \forall l$ 
and set
\[
t=K/\sqrt{n},\quad K>0. 
\]
Then we obtain 
\beq
\label{hoeff2}
\IP\lt[n^{-1}\lt|S_n -\IE S_n\rt|\geq K/\sqrt{n}\rt]
&\leq& 2e^{-{K^2/2}}\\
\IP\lt[n^{-1}\lt|T_n -\IE T_n\rt|\geq K/\sqrt{n}\rt]
&\leq& 2e^{-{K^2/2}}\label{hoeff22}.
\eeq
Note that the quantities $S_n, T_n$ depend on $x_i-x_j, y_i-y_j$, i.e.
\[
S_n=S_n(x_i-x_j, y_i-y_j), \quad T_n=T_n(x_i-x_j, y_i-y_j). 
\]
We use (\ref{hoeff2})-(\ref{hoeff22}) and the union bound to obtain
\beq
\label{10.1}&&{\IP\lt[\max_{i\neq j}n^{-1}\lt|S_n(x_i-x_j, y_i-y_j)-\IE S_n(x_i-x_j, y_i-y_j)\rt|\geq K/\sqrt{n}\rt]}\\
&&\hspace{1cm} \leq 2(m-1)e^{-K^2/2}\nn\\
&&{\IP\lt[\max_{i\neq j}n^{-1}\lt|T_n(x_i-x_j, y_i-y_j)-\IE T_n(x_i- x_j, y_i-y_j)\rt|\geq K/\sqrt{n}\rt]} \label{10.2}\\
&&\hspace{1cm}\leq 2(m-1)e^{-K^2/2}\nn
\eeq

Hence, 
 if (\ref{mesh}) holds 
for  any small number $\delta>0$,
then the right hand side of (\ref{10.1})-(\ref{10.2})
is less than $\delta$. 

\commentout{
Let 
\[
\mu_{ij}=\lt|\sum_{l=1}^n A^*_{li}(\om)A_{lj}(\om)\rt|
\]
and 
define the random variables $X_{ij}^l, Y_{ij}^l, l=1,...,n$, as 
\beq
\label{9.1}
X_{ij}^l&=&\cos{\lt[(\xi_l(x_i-x_j)+\eta_l(y_i-y_j))\om/z_0\rt]}\\
Y_{ij}^l&=&\sin{\lt[(\xi_l(x_i-x_j)+\eta_l(y_i-y_j))\om/z_0\rt]}\label{9.2}
\eeq
and their respective sums
\beqn
\label{9}
S^n_{ij}=\sum_{l=1}^n X^l_{ij},\quad T^{n}_{ij}=\sum_{l=1}^n Y^l_{ij}.
\eeqn
Then the absolute value of the
right hand side of (\ref{1.20}) is bounded by
\beq
\label{sum}
{1\over n} \lt|S_{ij}^n+iT^n_{ij}\rt|\leq {1\over n}\lt(\lt|S^n_{ij}-\IE S^n_{ij}\rt|
+\lt|T^n_{ij}-\IE T^n_{ij}\rt|+\lt|\IE (S^n_{ij}+iT^n_{ij})\rt|
\rt).
\eeq

To estimate the right hand side of  (\ref{sum}), we recall the Hoeffding inequality
\cite{Hoe}

\begin{proposition}
Let $X_1, ..., X_n$ be independent random variables. Assume
that $X_l \in [a_l, b_l], l=1,...,n$ almost surely.
Then for $S_n=\sum_{l=1}^nX_l$ we have
\beq
\label{hoeff}
\IP\lt[\lt|S_n -\IE S_n\rt|\geq nt\rt]
\leq 2\exp{\lt[-{2n^2 s^2\over \sum_{l=1}^n (b_l-a_l)^2}\rt]}
\eeq
for all positive values of $t$. 
\end{proposition}
We apply the Hoeffding inequality to both $S^n_{ij}$ and $T^n_{ij}$. To
this end, we have $a_l=-1, b_l=1, \forall l$ 
and set
\[
t=K/\sqrt{n},\quad K>0. 
\]
Then we obtain 
\beq
\label{hoeff2}
\IP\lt[n^{-1}\lt|S^n_{ij} -\IE S^n_{ij}\rt|\geq K/\sqrt{n}\rt]
&\leq& 2e^{-{K^2/2}}\\
\IP\lt[n^{-1}\lt|T^n_{ij}-\IE T^n_{ij}\rt|\geq K/\sqrt{n}\rt]
&\leq& 2e^{-{K^2/2}}
\eeq
which means that the event  
\[
{1\over n}\lt(\lt|S^n_{ij}-\IE S^n_{ij}\rt|
+\lt|T^n_{ij}-\IE T^n_{ij}\rt|\rt)\geq 2K/\sqrt{n},\quad n\to \infty
\]
has a super-exponentially small probability for large $K$. 
}

The third term on the right hand side of (\ref{sum})
can be calculated as follows. By the mutual  independence
of $\xi_l$ and $\eta_l$ we have
\beqn
{1\over n}\lt|\IE (S_n+iT_n)\rt|&=&{1\over n}\lt|\sum_{l=1}^n\IE(X_l+iY_l)\rt|\\
&=&
{1\over n}\lt|\sum_{l=1}^n \IE \lt(e^{i\xi_l\om(x_i-x_j)/z_0}\rt)
\IE\lt(e^{i\eta_l\om(y_i-y_j)/z_0}\rt)\rt|\\
&=&\lt|\IE \lt(e^{i\xi_l\om(x_i-x_j)/z_0}\rt)
\IE\lt(e^{i\eta_l\om(y_i-y_j)/z_0}\rt)\rt|
\eeqn
since $\xi_l, \eta_l, l=1,...,n$ are independently identically
distributed.

Simple calculation with the uniform distribution on $[0,A]\times [0,A]$
yields
\beq
\lt|\IE \lt(e^{i\xi_l\om(x_i-x_j)/z_0}\rt)
\IE\lt(e^{i\eta_l\om(y_i-y_j)/z_0}\rt)\rt| 
&=&\lt|{e^{i\phi_{ij}}-1\over \phi_{ij} } \rt|
\lt|{e^{i\psi_{ij}}-1\over \psi_{ij} } \rt|\nn \\
&=& 4\lt|{\sin{{\phi_{ij}}\over 2}\over \phi_{ij}}\rt|\lt|
{\sin{\psi_{ij}\over 2}\over \psi_{ij}}\rt| \label{20.2}
\eeq
with
\[
\phi_{ij}=A\om (x_i-x_j)/z_0,\quad \psi_{ij}
=A\om (y_i-y_j)/z_0. 
\]
The optimal condition is to choose $A$ such that
\beq
\label{square-grid}
\phi_{ij}=\psi_{ij}\in 2\pi\IZ, 
\eeq
under which (\ref{20.2}) vanishes. 
Condition (\ref{square-grid})
can be fulfilled for an equally spaced grid as is assumed here.
Let
\[
\ell=\min_{i\neq j} |x_i-x_j|=\min_{i\neq j} |y_i-y_j|.
\]
The smallest $\ell$ satisfying condition (\ref{square-grid}) is given
by 
\beq
\label{mesh-size}
\ell={z_0\lambda\over A},\quad\lambda=2\pi/\om
\eeq
which can be interpreted as the  resolution  of the imaging system and 
 is equivalent  to the classical Rayleigh criterion. 

In this case, $\IE (S_n+iT_n)=0$ and hence
\[
\mu(\bA)\leq \sqrt{2}K/\sqrt{n}
\]
with probability $(1-\delta)^2$ under the condition (\ref{mesh}).

\commentout{
Let
\[
\ell_x=\min_{i\neq j} |x_i-x_j|,\quad
\ell_y=\min_{i\neq j} |y_i-y_j|
\]
and
\[
\phi=A\om \ell_x /z_0,\quad \psi=A\om \ell_y/z_0.
\]
To ensure (\ref{20.2}) sufficiently small, we 
set  
\beqn
{4\over \phi \psi}= \epsilon \ll 1
\eeqn
which   implies that for 
\beq
\label{res}
\sqrt{\ell_x\ell_y} = {2 z_0\over A \om \sqrt{\ep}}
\eeq
the pairwise coherence 
\[
\mu_{ij}=\lt|\sum_{l=1}^n A^*_{li}A_{lj}\rt|
\]
 can be bounded as
\beq
\label{deco}
\mu_{ij}\leq \epsilon+{\sqrt{2}K\over \sqrt{n}}
\eeq
with probability greater than $(1-2e^{-K^2/2})^2$.
}

\commentout{
This suggests that with fixed resolution $\ell_x, \ell_y$ $\mu(\bA)$ can be made arbitrarily small with
high probability 
by increasing the number of antennas and the aperture.
Condition (\ref{res}), however,  means that the mesh size, which
can be interpreted  as the spatial resolution, 
deteriorates from the classical Rayleigh formula
 by the factor $\ep^{-1/2}$. 
 }
\commentout{
\beq
\label{res}
{4\over \phi \psi}\sim {2K\over \sqrt{n}}
\eeq
so that all three terms on the right hand side of (\ref{sum})
have the same order of magnitude
\[
{K\over \sqrt{n}}
\]
 with probability larger
than $(1-2e^{-K^2/2})^2$. 
The condition (\ref{res})  implies that
\beq
\label{aper}
A \sim {{n}^{1/4}\over K^{1/2}} \lt({2\over \ell_x \ell_y}\rt)^{1/2}{z_0\over \om}
\eeq
under which 
\beq
\label{deco}
\mu_{ij}\leq {3K\over \sqrt{n}}
\eeq
with probability greater than $(1-2e^{-K^2/2})^2$.
}
\end{proof}

\subsection{Proof of Theorem \ref{lemma2}: spectral norm bound}\label{sec3.2}
\begin{proof}
For the proof, it suffices to show that
the matrix $\bA$ satisfies
\beq
\label{gram}
\|{n\over m}\bA \bA^*-\bI_n\|_2<1
\eeq
where $\bI_n$ is the $n\times n$ identity matrix
with the corresponding probability bound. 
By the Gershgorin circle theorem, (\ref{gram}) would in turn
follow from 
\beq
\label{19}
\mu \lt(\sqrt{n\over m}\bA^*\rt)< {1\over n-1}
\eeq
since the diagonal elements of ${n\over m}\bA \bA^*$ are
unity.

 Since $(\xi_i,\eta_i), i=1,...,n$ are uniformly distributed
 in $[0,A]\times [0,A]$, $\xi_i\neq \xi_j, \eta_i\neq \eta_j$
 with probability one. 
 
Summing over $\br_l, l=1,...,m$ we obtain 
\beq
\nn
{n\over m}\sum_{l=1}^m A_{jl}A^*_{li}
&=&{1\over m}  e^{i\om(\xi_{j}^2+\eta_{j}^2-\xi_i^2-\eta_i^2)/(2z_0)}{e^{i\om(\xi_i-\xi_j)(x_1+\sqrt{m}\ell)/z_0}-e^{i\om(\xi_i-\xi_j)x_1/z_0}\over 1-e^{i\om(\xi_i-\xi_j)\ell/z_0}}\\
&& \times
{e^{i\om(\eta_i-\eta_j)(y_1+\sqrt{m}\ell)/z_0}-e^{i\om(\eta_i-\eta_j)y_1/z_0}\over 1-e^{i\om(\eta_i-\eta_j)\ell/z_0}}.
\label{1.20'}
\eeq
Thus, 
\beq
\label{1.25}{n\over m} \lt|\sum_{l=1}^m A_{jl}A^*_{li}
\rt|
&\leq &{1\over m}
\lt|{\sin{\sqrt{m}\om (\xi_i-\xi_j)\ell\over 2 z_0}\over \sin{\om(\xi_i-\xi_j)\ell\over 2 z_0}}\rt|\lt|{\sin{\sqrt{m}\om (\eta_i-\eta_j)\ell\over 2 z_0}\over \sin{\om(\eta_i-\eta_j)\ell\over 2 z_0}}\rt|,
\eeq
where we have used 
 the identity
\beq
\label{4}
{ \lt|1-e^{i\theta}\rt|}={2\lt|\sin{\theta\over 2}\rt|}.
\eeq

Let 
\beq
\label{view}
\kappa=\min_{i\neq j}\min_{k\in \IZ }\lt\{\lt|{ \ell (\xi_i-\xi_j)\over  \lambda z_0}-k\rt|, \lt|{ \ell (\eta_i-\eta_j)\over \lambda z_0}-k\rt|\rt\}\leq 1/2
\eeq
which is nonzero with probability one.
For $i\neq j$ the random variables
\[
{\ell (\xi_i-\xi_j)\over  \lambda z_0},  \quad { \ell (\eta_i-\eta_j)\over\lambda  z_0}
\]
have the symmetric triangular distribution supported
on $[- \rho^{-1}, \rho^{-1}]$  with 
height $\rho=\lambda z_0/( \ell A)$.  Note that 
 $\rho^{-1}$ is an integer 
   by the choice (\ref{square-grid}). 
Hence the probability that $\{\kappa >\alpha\}$ for small $\alpha>0$ is larger than 
\[
(1-2\rho\alpha)^{n(n-1)} >1-2\rho n(n-1)\alpha,\quad
\rho={\lambda z_0\over \ell A}
\] 
where the power $n(n-1)$ accounts for  the number of
different pairs of random variables involved in
(\ref{view}). 

Using the inequality  that
\[
\sin{\pi\kappa} > 2\kappa,\quad \kappa\in (0,1/2), 
\]
(\ref{1.25})  and
the choice
\[
{1\over 2} \sqrt{n-1\over{ m}}=\alpha
\]
we deduce  
with probability larger than 
\[
1-2\rho n(n-1)\alpha=1-{\rho n(n-1)^{3/2}\over m^{1/2}}
\]
the decoherence  estimate 
\[
\mu\lt(\sqrt{n\over m}\bA^*\rt)< {1\over 4 m \alpha^2}
\]
 implying  (\ref{19}). 
 \end{proof}

\section{Inverse  Born scattering}
\label{sec:ibs}
In this section, we consider two imaging settings where the targets
are scatterers instead of sources under the Born approximation
(\ref{born}). 
   
\subsection{Response matrix (RM) imaging}\label{sec:rm}
For the coherence calculation, we have
\beq\label{35-2}
\sum_{l=1}^{n^2} A^{\rm RM*}_{li}A^{\rm RM}_{lj}
&=&\sum_{p,q=1}^n G(\ba_p,\br_j)G(\br_j,\ba_q)G^*(\ba_p,\br_i)G^*(\br_i,\ba_q)\\
&=&\lt[\sum_{p=1}^n G(\ba_p,\br_j)G^*(\ba_p,\br_i)\rt]^2\nn
\eeq
and thus
\[
\mu\lt(\bA^{\rm RM}\rt)=\mu^2(\bA).
\]
In view of (\ref{35-2})
and Theorem \ref{thm3} the following theorem is automatic.
\begin{theorem}
\label{thm5}
Under  the assumptions (\ref{aperture}) and (\ref{mesh}) 
the coherence  of $\bA$ satisfies
\[
\mu(\bA^{\rm RM})\leq 2K^2/{n}
\]
with probability greater than  $(1-\delta)^2$. 
\end{theorem}

We now proceed to establish the counterpart of Theorem \ref{lemma2}. 
\begin{theorem}
\label{thm6}
The matrix $\bA^{\rm RM}$ has full rank and 
its spectral norm satisfies the bound
\beq
\label{norm6}
\| \bA^{\rm RM}\|_2^2\leq 2m/n^2
\eeq
 with probability greater than  or equal to 
\[
1-{\rho n^2(n^2-1)^{3/2}\over m^{1/2}},\quad \rho={\lambda z_0\over \ell A}.
\] 
\end{theorem}
\begin{remark}\label{rmk3}
As in Remark \ref{rmk2}, 
Theorems \ref{thm5} and \ref{thm6} imply  that
with probability  greater than
\[
1-2\delta- {\rho n^2(n^2-1)^{3/2}\over m^{1/2}}
\]
of the sensor ensemble, {\bf all} targets of sparsity
\[
s<{1\over 2} (1+{{n}\over {2} K^2})
\]
can be recovered exactly by BP as well as OMP.

\end{remark}

\begin{proof} We proceed as in the proof of Theorem \ref{lemma2}. As before, we seek  to prove
\beq
\label{19'}
\mu \lt({n\over \sqrt{ m}}\bA^{\rm RM*}\rt)< {1\over n^2-1}. 
\eeq

For the RM setting, (\ref{1.20}) becomes
\beq
\nn
{n^2\over m}\sum_{j=1}^m A^{\rm RM}_{jl}A^{\rm RM*}_{jl'}
&=&{1\over m}  e^{i\om(\xi_{k}^2+\eta_{k}^2+\xi_i^2+\eta_i^2-\xi_{k'}^2-\eta_{k'}^2-\xi_{i'}^2-\eta_{i'}^2)/(2z_0)}\\
&&\times {e^{i\om(\xi_i+\xi_{k}-\xi_{i'}-\xi_{k'})(x_1+\sqrt{m}\ell)/z_0}-e^{i\om(\xi_i+\xi_{k}-\xi_{i'}-\xi_{k'})x_1/z_0}\over 1-e^{i\om(\xi_i+\xi_{k}-\xi_{i'}-\xi_{k'})\ell/z_0}}\nn\\
&& \times
{e^{i\om(\eta_i+\eta_{k}-\eta_{i'}-\eta_{k'})(y_1+\sqrt{m}\ell)/z_0}-e^{i\om(\eta_i+\eta_k-\eta_{i'}-\eta_{k'})y_1/z_0}\over 1-e^{i\om(\eta_i+\eta_k-\eta_{i'}-\eta_{k'})\ell/z_0}}
\label{2.20}
\eeq
where $l=i(n-1)+k, l'=i'(n-1)+k'.$

We apply the same analysis as  (\ref{1.20}) here. 
Let 
\beq
\label{view2}
\kappa=\min_{l\neq l'}\min_{k\in \IZ }\lt\{\lt|{ \ell (\xi_i+\xi_k-\xi_{i'}-\xi_{k'})\over  \lambda z_0}-k\rt|, \lt|{ \ell (\eta_i+\eta_k-\eta_{i'}-\eta_{k'})\over \lambda z_0}-k\rt|\rt\}.
\eeq
which is nonzero with probability one.
For $l\neq l'$ the probability density functions  (PDF) for the random variables
\[
{ \ell (\xi_i+\xi_k-\xi_{i'}-\xi_{k'})\over  \lambda z_0},\quad
{ \ell (\eta_i+\eta_k-\eta_{i'}-\eta_{k'})\over \lambda z_0}
\]
are either  the symmetric triangular distribution or
its self-convolution supported
on $[- 2\rho^{-1},  2\rho^{-1}]$. 
In either case, their PDFs are bounded by $\rho$ (indeed, by $2\rho/3$). 
Hence the probability that $\{\kappa >\alpha\}$ for small $\alpha>0$ is larger than 
\[
(1-2\rho\alpha)^{n^2(n^2-1)} >1-2\rho n^2(n^2-1).
\] 
With the choice
\[
{1\over 2} \sqrt{n^2-1\over{ m}}=\alpha
\]
we deduce that 
\[
\mu\lt(\sqrt{n^2\over m}\bA^{\rm RM*}\rt)<{1\over n^2-1}
\]
with probability larger than 
\[
1-2\rho n^2(n^2-1)\alpha=1-{\rho n^2(n^2-1)^{3/2}\over m^{1/2}}. 
\]

\end{proof}

As before, the above estimates yield
the following result.

\begin{theorem}
\label{thm:rm}
Consider the response matrix imaging with the target vector  randomly drawn from
the target ensemble and the antenna array randomly
drawn from the sensor ensemble.  If (\ref{aperture}) and
(\ref{m-1}) hold
then
 the targets of sparsity up to 
\beq
\label{spark2}
{n^2\over 64 \ln{2m\over \delta}\ln{m\over \ep}}
\eeq
can be recovered exactly by BP with
probability greater than or equal to (\ref{prob}).

\end{theorem}
\subsection{Synthetic aperture (SA) imaging}\label{sec:sa}
In view of (\ref{51}), 
we obtain
\[
\mu \lt(\bA^{\rm SA}(\om)\rt)=\mu(\bA(2\om)).
\]

The following
result is an immediate consequence of
 the  correspondence (\ref{50})-(\ref{51}) between SA imaging and inverse source setting. 
\begin{theorem}
\label{thm-sa}
 Let the target vector be randomly drawn from
the target ensemble and the antenna array be randomly
drawn from the sensor ensemble.  If
\beq
\label{49-3}
{2\over \rho} \in \IN
\eeq
and (\ref{m-1}) hold 
then the targets of sparsity up to
\beq
\label{spark3}
{n\over 64 \ln{2m\over \delta}\ln{m\over \ep}}
\eeq
can be recovered exactly by BP with
probability greater than or equal to (\ref{prob}). 
\end{theorem}

\begin{remark}\label{rmk4}
As in Remark \ref{rmk2}, conditions  (\ref{mesh}) and (\ref{49-3})  imply that 
with probability  greater than
\[
1-2\delta- {\rho n(n-1)^{3/2}\over m^{1/2}}
\]
of the sensor ensemble, {\bf all} targets of sparsity
\[
s<{1\over 2} (1+{\sqrt{n}\over \sqrt{2} K})
\]
can be recovered exactly by BP as well as OMP.

\end{remark}

\section{Numerical simulations }\label{sec:num}
In the simulations, we set $z_0=10000$ and for the most
part $\lambda=0.1$ to enforce the second condition of
the paraxial regime (\ref{6-2}). The computational domain is 
$[-250, 250]\times [-250, 250]$ with mesh-size $\ell=10$. 
The threshold, optimal aperture according to Theorem~\ref{thm3} is $A=100$. As a result, the first condition of
the paraxial regime (\ref{6-2}) is also enforced. 
Note that the Fresnel number for this setting is
\[
{(A+L)^2\over z_0\lambda}=360\gg 1
\]
indicating that this is not the Fraunhofer  diffraction regime
and the Fourier approximation of the paraxial Green function is
not appropriate  \cite{BW}.

We use the 
true Green
function 
(\ref{green}) in the direct simulations
and its paraxial approximation for inversion. In other words, we allow model 
 mismatch between the propagation  and inversion steps.
 The degradation in performance can be seen in the figures 
 but is still manageable as the simulations are firmly in
 the Fresnel diffraction regime. 
 The stability of BP  with  linear model mismatch
 has been analyzed in~\cite{HS09} for the case when the matrix satisfies the Restricted  Isometry Property (RIP), see
 Appendix \ref{app:rip}.

 In the left plot of Figure \ref{fig1}, the coherence  is calculated
 with aperture $A\in [10, 200]$ and $n=100$ for
 the sensing matrices  with the exact Green function (red-solid curve) as entries 
 and its paraxial approximation (black-asterisk  curve).  The coherence of
 the exact sensing matrix at the borderline of the paraxial regime with $z_0=1000, \lambda=1$ is also calculated (blue-dashed  curve). All three curves track one another closely and 
 flatten near and beyond  $A=100$ in agreement with
 the theory (Theorem \ref{thm3}), indicating
 the validity of the optimal aperture throughout the
 paraxial regime. 
 
 Figure \ref{fig1} (right plot)  displays the numerically found maximum number
 of recoverable source points as a function of $n$ with $A=100$ by using the exact (red-solid curve) and paraxial (black-asterisk curve) sensing matrices.  The maximum number of recoverable targets (MNRT) is in principle a random variable 
 as our theory is formulated in terms of the target
 and sensor ensembles. To compute MNRT,
 we start with one target point and  apply  the sensing scheme.
 If the recovery is (nearly) perfect a new target vector with one additional support is randomly drawn and the sensing scheme is
rerun. We iterate this process until the sensing scheme fails
to recover the targets and then we record the target support in
the previous iterate as MNRT. {\em This is an one-trial test and
no averaging is applied.} 
The linear profile in the right plot of Figure \ref{fig1} is  consistent with the prediction (\ref{spark}) of Theorem \ref{thm-is}. 

  \begin{figure}[t]
\begin{center}
\includegraphics[width=0.44\textwidth]{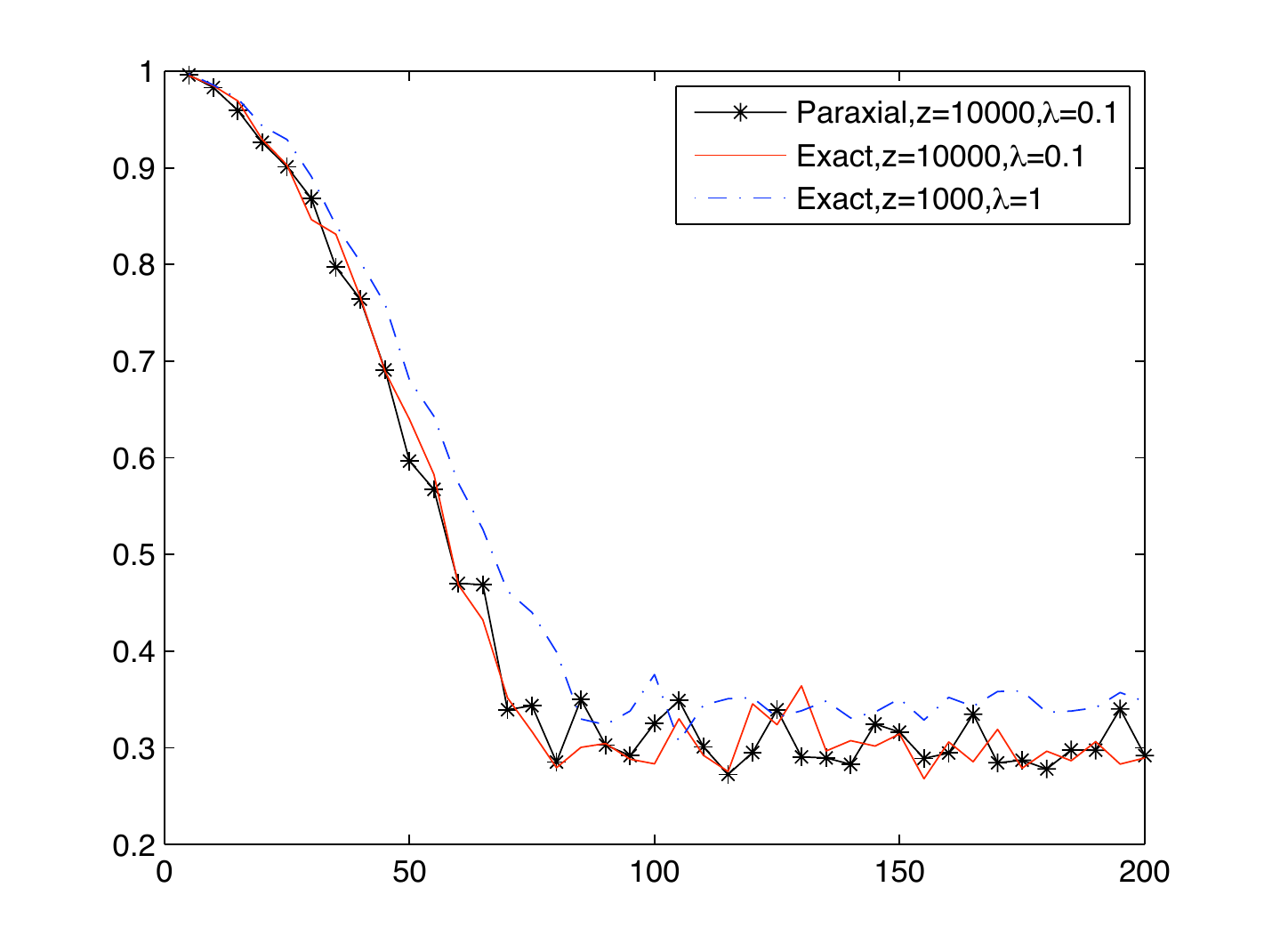}
  \includegraphics[width=0.44\textwidth]{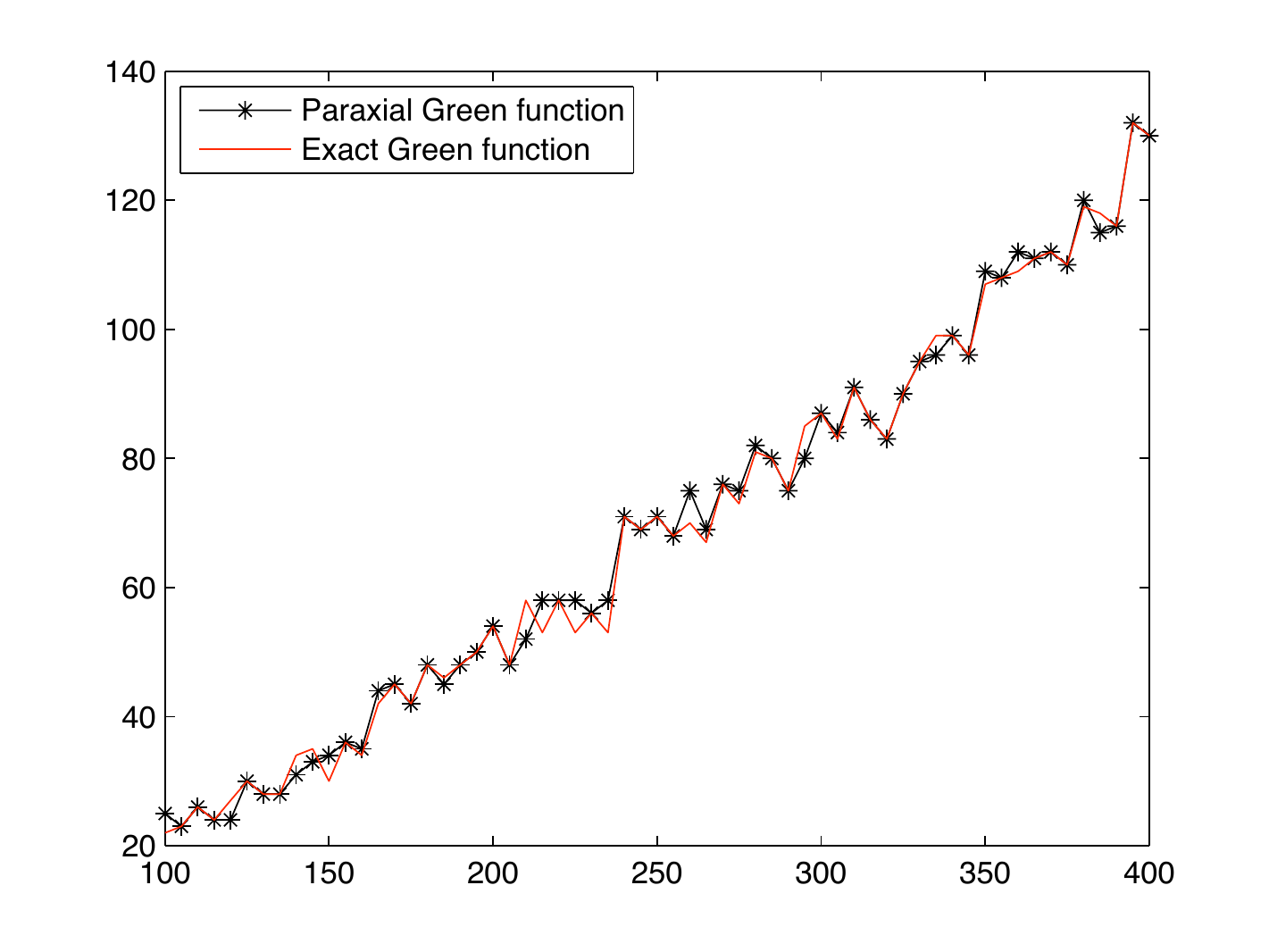}
\end{center}
\caption{(Left) The red-solid  and black-asterisk curves are, respectively,  the 
coherence for the exact and paraxial sensing matrices
for $z_0=10000, \lambda=0.1, n=100$  as a function
of aperture. The blue-dashed  curve is the coherence for the exact
sensing matrix for $z_0=1000, \lambda=1, n=100$; (Right) the empirical, maximum number of recoverable sources  with
 $|\sigma|=1$ v.s. the number $n$
 of antennas for $A=100, \lambda=1$ by using the paraxial (black-asterisk)
 and exact (red-solid) sensing matrices. }
 \label{fig1}
 \end{figure}
 
 To reduce the computational complexity of the compressed
 sensing step, we use an iterative scheme called
 {\em Subspace Pursuit} (SP) \cite{DM}. It has been shown
 to yield the BP solution under the RIP  \cite{DM} (see
 Appendix \ref{app:rip}.

In the {scattering simulation}, we use the Foldy-Lax formulation
 accounting for all the multiple scattering effect \cite{mfirm-ip}.
 Hence there are two mismatches   (the paraxial approximation
 and the Born approximation) in the simulation. 

 In the left plot of Figure \ref{fig-scatt2} the compressed sensing image with
 RM set-up is shown
 for $A=100$ and $n=20$. The size of the sensing matrix is
 $400\times 2500$ and $35$ targets are (nearly) exactly recovered. 
 For comparison, the image obtained by
 the linear processor of 
 the traditional matched field processing  is  
 shown on the right. In Appendix \ref{sec:A},
 we outline the rudiments  of matched field processing. 
 \begin{figure}[t]
\begin{center}
\includegraphics[width=0.47\textwidth]{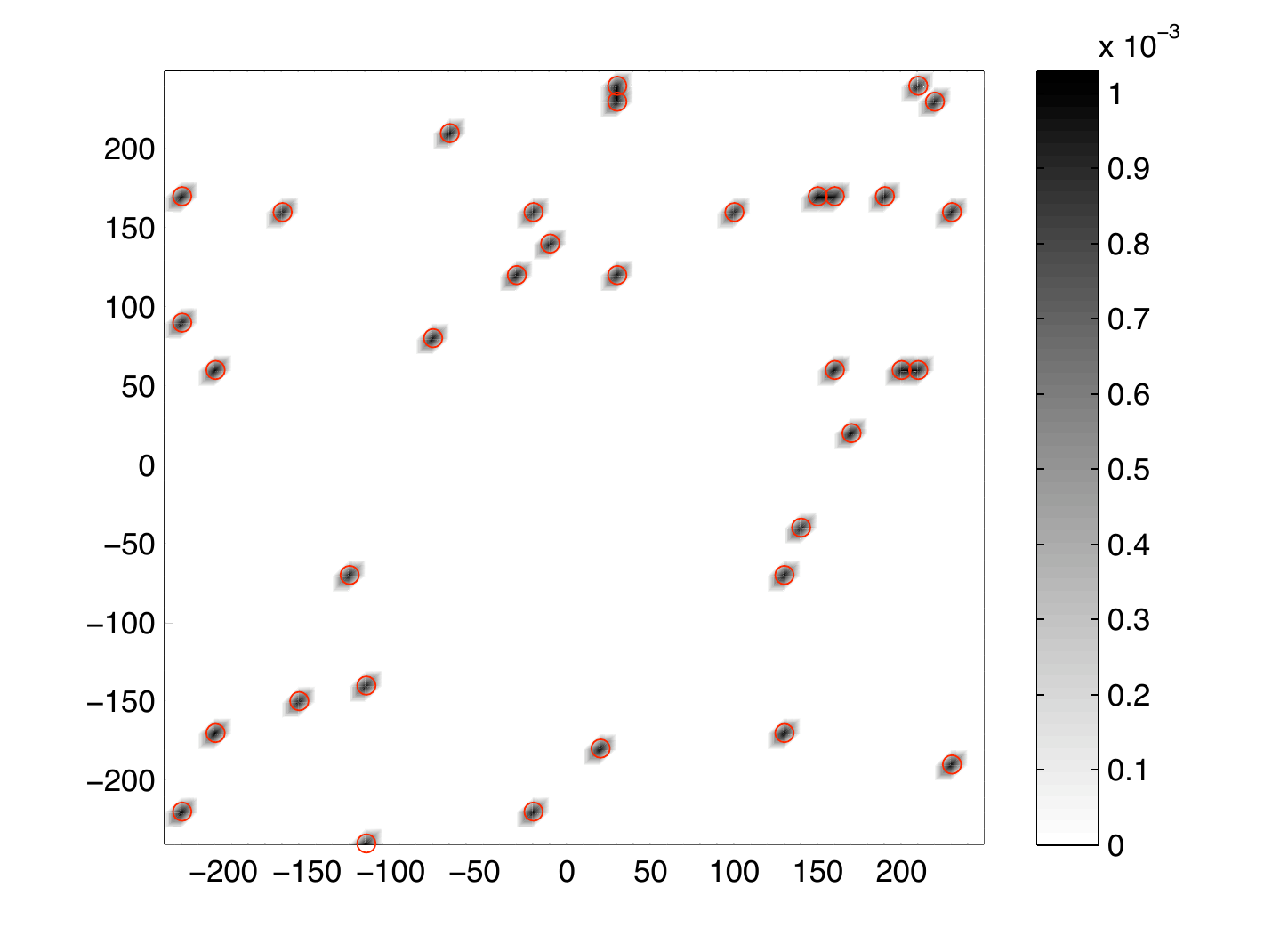}
\includegraphics[width=0.47\textwidth]{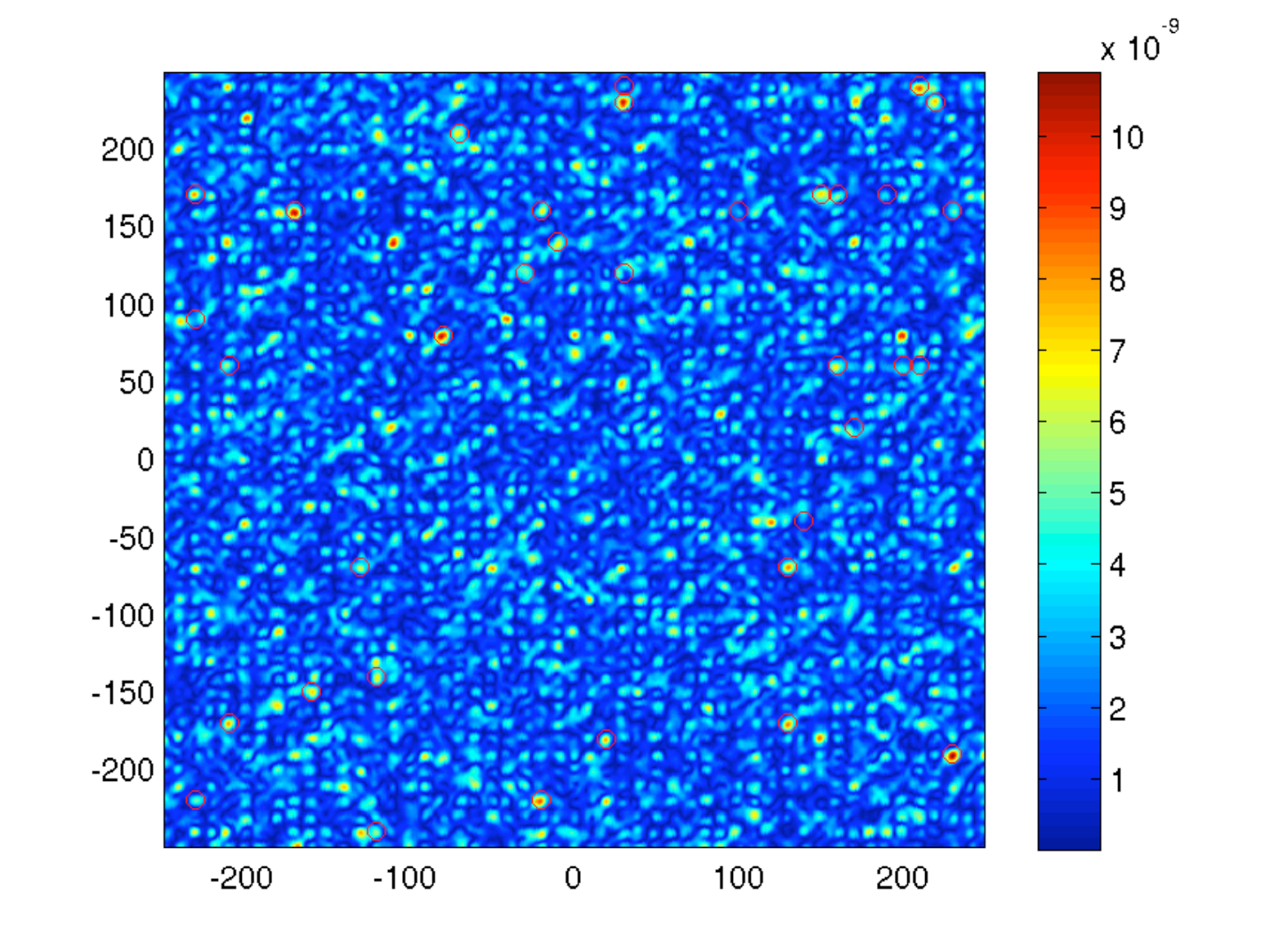}
\end{center}
 \caption{(Left) 35 scatterers are perfectly recovered by
 compressed sensing technique
 with $20$ antennas.    The red circles represent the true locations of the targets. The plot on the right is 
 produced by the conventional matched field processing. }
 \label{fig-scatt2}
 \end{figure}

In Figure \ref{fig-scatt} 
the numerically found maximum number of recoverable scatterers is depicted as
a function of the number of antennas  for $A=100$
and  for both RM and SA imaging set-ups by using
the paraxial and exact sensing matrices.  Clearly,
both curves are qualitatively consistent with the predictions
(\ref{spark2}) and (\ref{spark3}).

\begin{figure}[t]
\begin{center}
\includegraphics[width=0.45\textwidth]{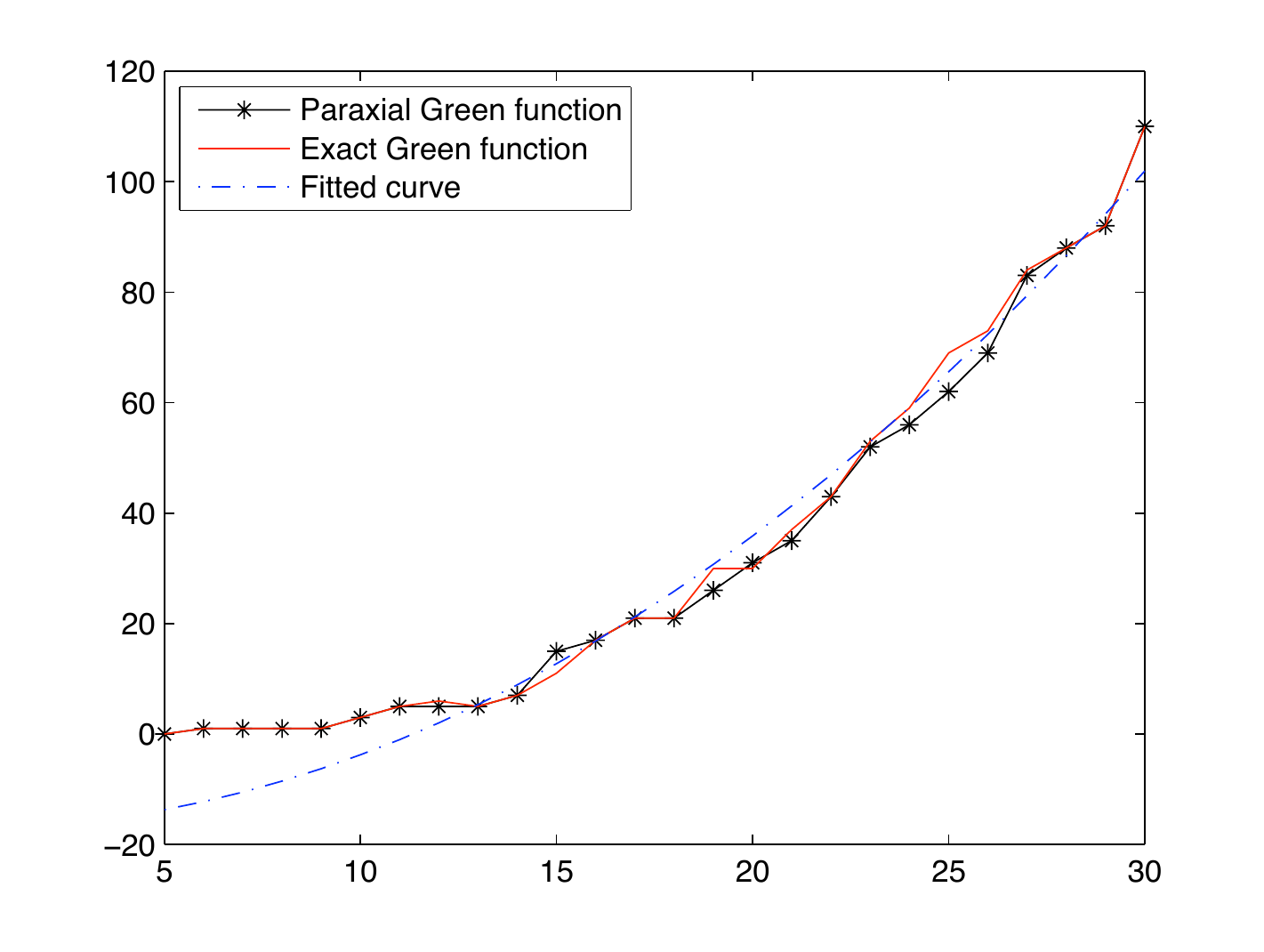}
\includegraphics[width=0.45\textwidth]{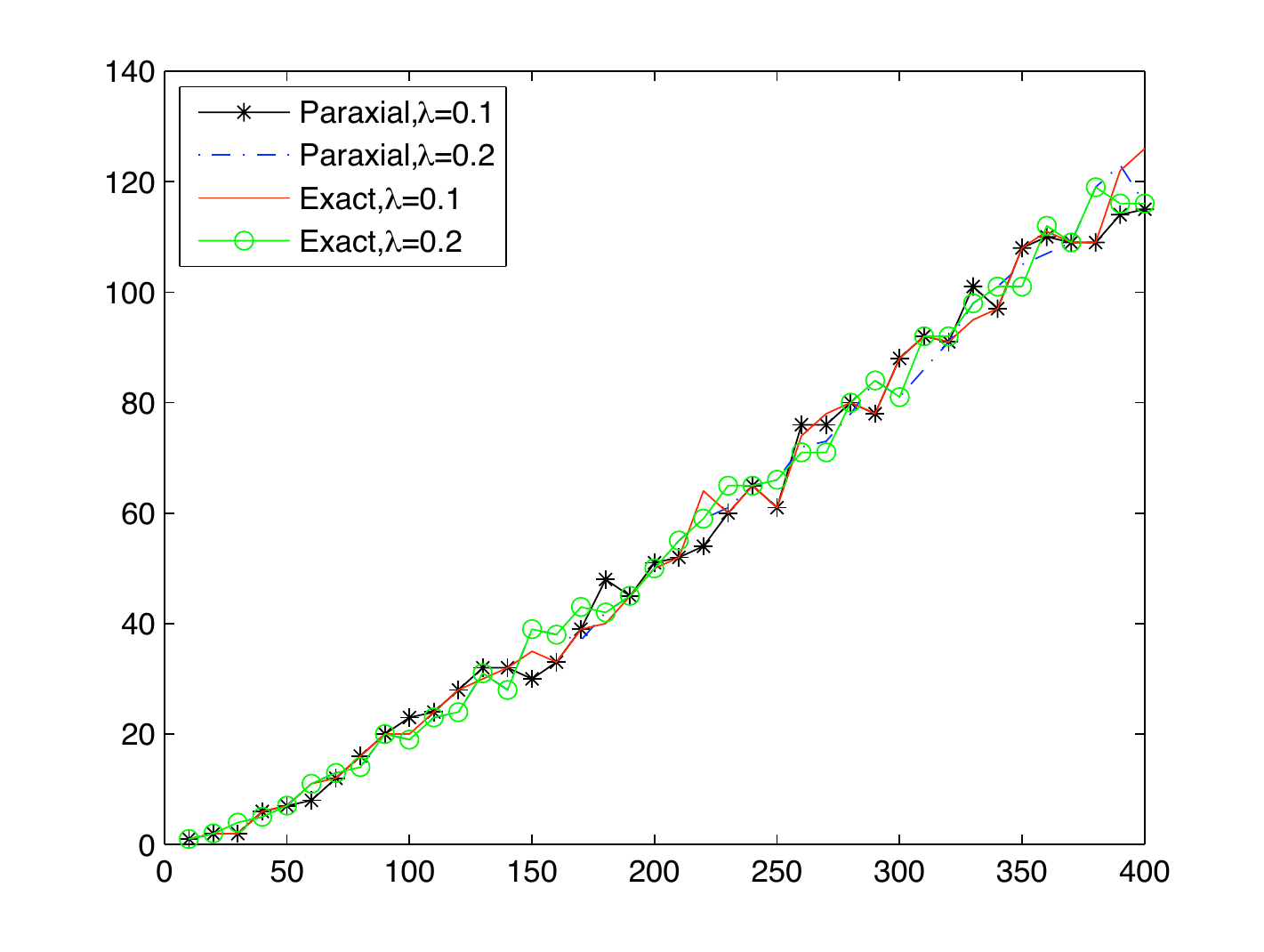}
\end{center}
 \caption{The empirical maximum number of recoverable 
 scatterers  (left for RM, right for SA) with $|\sigma|=0.001$ v.s. the number $n$
 of antennas (or antenna locations)  for $A=100$. 
 The data for $n\in [10,30]$ in the RM plot is fitted with
  the parabola (blue-dashed curve): $-16.4950+0.1366*x^2$. 
  The wavelength is $0.1$ for the RM case. 
  The SA plot depicts   the number of recoverable
  scatterers in four settings: paraxial sensing matrix with 
  $\lambda=0.1$ (black-asterisk), paraxial sensing matrix with $\lambda=0.2$ (blue-dashed), exact sensing matrix with $\lambda=0.1$ (red-solid) and exact sensing matrix with $\lambda=0.2$ (green-circled) }
 \label{fig-scatt}
 \end{figure}
 \commentout{
\begin{figure}[t]
\begin{center}
\includegraphics[width=0.45\textwidth]{1f-figures/paraxial/RM_detectable_curve.pdf}
\includegraphics[width=0.45\textwidth]{1f-figures/paraxial/SA_detectable.pdf}
\end{center}
 \caption{The blue curves are the empirical maximum  number of recoverable 
 scatterers  (left for RM, right for SA) with $|\sigma|=0.001$ v.s. the number $n$
 of antennas (or antenna locations)  for $A=100, \lambda=0.1$. 
 The data for $n\in [10,30]$ in the RM plot (left) is fitted with
  the parabola (red): $-14.9646+0.1486 n^2$. The red curve in
  the right plot is the number of recoverable
  scatterers for $\lambda=0.2$ with all other parameters the same as for the blue curve.  }
 \label{fig-scatt}
 \end{figure}
 }

 \commentout{
\section{Numerical simulations }
In the simulations, we set $z_0=1000$ and for the most
part $\lambda=1$. We use the 
true Green
function used in the simulation
\[
G(\br,\ba)={e^{i\om|\br-\ba|}\over 4\pi |\br-\ba|}
\]
instead of its paraxial approximation. In other words, there
is a mismatch between the sensing matrix and the
true propagator. 

For $\lambda=1$  we use the same computational domain. 
The condition (\ref{aperture}) in  Theorem~\ref{thm3} corresponds to  $A=100k, k\in \IN$. 

 In the left plot of Figure \ref{fig1}, the coherence  is calculated
 with aperture $A\in [10, 200]$ and $n=100$. The curve
 flattens near and beyond  $A=100$ in agreement with
 the theory. 
 The optimal incoherence
 is achieved at around $A=100$, consistent with the paraxial  theory. The right plot displays the numerically found maximum number
 of recoverable source points as a function of $n$ with $A=100$.
 The linear profile is also consistent with the prediction
 of the paraxial theory. 
  \begin{figure}[t]
\begin{center}
\includegraphics[width=0.44\textwidth]{1f-figures/paraxial/coherence.pdf}
  \includegraphics[width=0.44\textwidth]{1f-figures/paraxial/Recoverable_sources.pdf}
\end{center}
\caption{(Left) Coherence  as a function
of aperture  with 100 antennas; (Right) the empirical maximum number of recoverable sources  with
 $|\sigma|=1$ v.s. the number $n$
 of antennas for $A=100, \lambda=1$.}
 \label{fig1}
 \end{figure}
 
 \begin{figure}[t]
\begin{center}
\includegraphics[width=0.44\textwidth]{1f-figures/1_1.pdf}
  \includegraphics[width=0.44\textwidth]{1f-figures/A100-pass.pdf}
\end{center}
\caption{(Left) Coherence  as a function
of aperture  with 100 antennas; (Right) the empirical maximum number of recoverable sources  with
 $|\sigma|=1$ v.s. the number $n$
 of antennas for $A=100, \lambda=1$.}
 \label{fig1-2}
 \end{figure}

 \commentout{
 \begin{figure}[t]
\begin{center}
\includegraphics[width=0.44\textwidth]{1f-figures/1_1.pdf}
\includegraphics[width=0.47\textwidth]{1f-figures/100to400.pdf}
\end{center}
 \caption{Coherence  as a function
 of aperture  with 100 antennas (left) and
the number $n$ of antennas with $A=100, 200, 300, 400$ (right).}
 \label{fig1}
 \label{fig100to400}
 \end{figure}
 }
 
\commentout{
In Figure \ref{fig4} the coherence  is calculated
for $A=\in [10, 200]$  with the antenna number $n=(1+A/10)^2$.
 \begin{figure}[t]
\begin{center}
\includegraphics[width=0.5\textwidth]{1f-figures/4.pdf}
\end{center}
 \caption{Coherence  as a function of
 $A$ with $n=(1+A/10)^2$.}
 \label{fig4}
 \end{figure}
 }

\commentout{
In Figure \ref{fig5} the number of recoverable source points
is plotted as a function of aperture with $n=(1+A/20)^2$. 

 \begin{figure}[t]
\begin{center}
\includegraphics[width=0.5\textwidth]{1f-figures/5.pdf}
\end{center}
 \caption{The number of recoverable source points as a function of
 $A$ with $n=(1+A/20)^2$.}
 \label{fig5}
 \end{figure}
 }

\commentout{
 \begin{figure}[t]
\begin{center}
\includegraphics[width=0.45\textwidth]{1f-figures/12.pdf}
\includegraphics[width=0.45\textwidth]{1f-figures/14.pdf}
\end{center}
 \caption{Compressed sensing image (left) of $138$ recovered  source points   with $A=100, n=400$. Matched field processing  image (right)
   with $A=400, n=400$  is
   obtained with the thresholding  at the intensity level 0.027.  
   The red circles represent the true locations of the targets.}
 \label{fig9}
 \label{fig10}
 \end{figure}
 }
\commentout{
 \begin{figure}[t]
\begin{center}
\includegraphics[width=0.45\textwidth]{1f-figures/10.pdf}
\includegraphics[width=0.45\textwidth]{1f-figures/11.pdf}
\end{center}
 \caption{Matched field processing  image (left)
   with $A=400$ with $n=441$. The image on the right is
   obtained when the thresholding  at the intensity level 0.033
   is applied to the image on the left. }
 \label{fig9}
 \end{figure}
 }

\subsection{Scattering simulation}
  Let the  randomly distributed point scatterers of strengths $\sigma_j$ be located
at $\bx_j, j=1,2,3,...s$. Then the resulting Green function  $\tilde G$ obeys  the Lippmann-Schwinger equation
\beqn
\tilde G(\bx,\ba_i)=G(\bx,\ba_i)
+\sum_{j=1}^s \sigma_j
 G(\bx,\bx_j)\tilde  G(\bx_j,\ba_i),\quad i=1,...,n.
\eeqn 
$\tilde G$ contains all the multiple scattering events
and can be calculated numerically (see e.g. \cite{mfirm-ip}
for details).  

 In the left plot of Figure \ref{fig-scatt2} the compressed sensing image with
 RM set-up is shown
 for $A=100, \lambda=0.1$ and $n=20$. 
 For comparison, the image obtained by
 the linear processor of 
 the traditional matched field processing theory is  
 shown on the right. In Appendix \ref{sec:A},
 we outline the rudiments  of matched field processing. 
 \begin{figure}[t]
\begin{center}
\includegraphics[width=0.47\textwidth]{1f-figures/paraxial/cs_imaging.pdf}
\includegraphics[width=0.47\textwidth]{1f-figures/paraxial/RM_imaging.pdf}
\end{center}
 \caption{(Left) 43 scatterers are perfectly recovered by
 compressed sensing technique
 with $20$ antennas.    The red circles represent the true locations of the targets. The plot on the right is 
 produced by the conventional matched field processing. }
 \label{fig-scatt2}
 \end{figure}

 \begin{figure}[t]
\begin{center}
\includegraphics[width=0.47\textwidth]{1f-figures/active20_wavelength1_cs.pdf}
\includegraphics[width=0.47\textwidth]{1f-figures/active20_wavelength1_MFP.pdf}
\end{center}
 \caption{(Left) 65 scatterers are perfectly recovered by
 compressed sensing technique
 with $20$ antennas.    The red circles represent the true locations of the targets. The plot on the right is 
 produced by the conventional matched field processing. }
 \label{fig-scatt2-2}
 \end{figure}

In Figure \ref{fig-scatt} 
the numerically found maximum number of recoverable scatterers  as
a function of the number of antennas for $A=100$
is shown for both RM and SA set-ups.  Clearly,
both curves are consistent with the prediction of the theory.

\commentout{
\begin{figure}[t]
\begin{center}
\includegraphics[width=0.6\textwidth]{1f-figures/A100-pass.pdf}
\end{center}
 \caption{The number of recoverable sources  with
 $|\sigma|=1$ v.s. the number $n$
 of antennas for $A=100, \lambda=1$. }
 \label{fig-passive}
 \end{figure}
}

\begin{figure}[t]
\begin{center}
\includegraphics[width=0.45\textwidth]{1f-figures/paraxial/RM_detectable_curve.pdf}
\includegraphics[width=0.45\textwidth]{1f-figures/paraxial/SA_detectable.pdf}
\end{center}
 \caption{The empirical maximum number of recoverable 
 scatterers  (left for RM, right for SA) with $|\sigma|=0.001$ v.s. the number $n$
 of antennas (or antenna locations)  for $A=100, \lambda=1$. 
 The data for $n\in [10,30]$ in the RM plot is fitted with
  the parabola (red): $-14.9646+0.1486 n^2$. }
 \label{fig-scatt}
 \end{figure}

\begin{figure}[t]
\begin{center}
\includegraphics[width=0.45\textwidth]{1f-figures/RM-n.pdf}
\includegraphics[width=0.45\textwidth]{1f-figures/SA.pdf}
\end{center}
 \caption{The empirical maximum number of recoverable 
 scatterers  (left for RM, right for SA) with $|\sigma|=0.001$ v.s. the number $n$
 of antennas (or antenna locations)  for $A=100, \lambda=1$. 
 The data for $n\in [10,30]$ in the RM plot is fitted with
  the parabola (red): $-12.6711+0.1985n^2$. }
 \label{fig-scatt-2}
 \end{figure}

\commentout{
 The sonar problem is often concerned with estimating the
 location of a point acoustic source in an underwater  environment. 
 The advantages of the linear processor are
 its simplicity and robustness to errors (mismatch)
 in the modeled fields; it has been shown
 that the linear processor is the least sensitive of all
 the processors to errors in the modeled fields.
 A disadvantage of the linear processor is the presence
 of side lobes \cite{BKM, Tol}.
 }

\commentout{

In this following numerical example, we consider
the free space whose Green function is well known. 
The computation domain is $[-5000, 5000]\times [0,5000]$.
 51 equally spaced antennas are placed at $x=-5000, y\in [0, 5000]$. The coherence length $\ell_c$ is estimated according to
 Rayleigh's resolution formula and used as the grid size
 for setting up the sensing matrix. 1000 
 point sources are randomly placed in the right half domain
 $[0,5000]\times [0, 5000]$. 20 different  frequencies 
 corresponding to equally spaced wavelength 
 from 52 to 90 are used. 
 
 \begin{figure}[t]
\begin{center}
\includegraphics[width=0.7\textwidth]{Passive_filter_100sources_x.pdf}
\end{center}
 \caption{MFP with the linear processor and 100 randomly distributed point sources of unit strength whose locations are marked
 by crosses. The appearance of
 side lobes increases ambiguities. }\label{MFP100}
 \end{figure}
  
 \begin{figure}[t]
\begin{center}
\includegraphics[width=0.7\textwidth]{Passive_OMP_100sources.pdf}
\end{center}
 \caption{A near perfect reconstruction with compressed sensing (which algorithm?).}\label{MP100}
 \end{figure}
 }

 \commentout{
  \begin{figure}[t]
\begin{center}
\includegraphics[width=0.7\textwidth]{Passive_OMP_clutterint_10_negdomain.pdf}
\end{center}
 \caption{A near perfect reconstruction with compressed sensing
 in the presence of 100 randomly distributed clutter particles of scattering strength
 10 in the left half domain $x\in [-5000, 0]$.}\label{MP100clutter}
 \end{figure}
 }
 
 \commentout{
 \begin{figure}[t]
\begin{center}
\includegraphics[width=0.7\textwidth]{Passive_filter_200sources_x.pdf}
\end{center}
 \caption{MFP with the linear processor and 200 randomly distributed point sources whose locations are marked
 by crosses. The appearance of
 side lobes increases ambiguities. }\label{MFP200}
 \end{figure}
 
  \begin{figure}[t]
\begin{center}
\includegraphics[width=0.7\textwidth]{Passive_OMP_200sources.pdf}
\end{center}
 \caption{CS with CoSaMP fails.}\label{MP200}
 \end{figure}
}
}

\section{Conclusions}
In this paper, we have studied the imaging problem
from the perspective of  compressed sensing, in particular
the idea that sufficient incoherence and sparsity guarantee
uniqueness of the solution. Moreover, by adopting 
the target ensemble following \cite{Tropp} and the sensor
ensemble, the 
 maximum number of recoverable targets is proved
to be at least  proportional to the number of measurement
data modulo a log-square factor with overwhelming probability. 

We have analyzed three imaging settings: the inverse source,
the inverse scattering with the response matrix and 
with the synthetic aperture. 
Important contributions of our analysis include  the discoveries
of the 
decoherence effect  induced  by random antenna locations and
the threshold aperture defined by
$\rho=1$ for source  and RM imaging
and  $\rho=1/2$ for SA imaging where
$\rho=\lambda z_0/(A\ell)$.  

In this paper we have considered the localization of point targets
and the determination of their amplitudes. A natural next step is to consider extended targets. However our approach does not extend in
a straightforward manner to imaging of extended targets, as can be easily seen. Assume that we model an extended target approximately as an ensemble of point targets that are spaced very close together. Clearly,
this requires the mesh size $\ell$ to be so small as to render  $\rho\gg 1$. To apply  our theorems  would then require that the aperture and the number of
antennas  increase without bound.
Clearly this is not a feasible way to image extended targets via compressed sensing. Therefore a somewhat different approach, on which we plan to report in our future work, is required for extended targets.

\appendix

\section{Restricted isometry property (RIP) }\label{app:rip}

A fundamental notion in compressed sensing under which
BP yields the unique exact solution is the restrictive isometry property due to Cand\`es and Tao \cite{CT}.
Precisely, let  the sparsity $s$  of the target vector be the
number of nonzero components of $X$ and define the restricted isometry constant $\delta_s$
to be the smallest positive number such that the inequality
\[
(1-\delta_s) \|Z\|_2^2\leq \|\bA Z\|_2^2\leq (1+\delta_s)
\|Z\|_2^2
\]
holds for all $Z\in \IC^m$ of sparsity at most $ s$.

Now we state the fundamental   result 
of the RIP approach. 
\begin{theorem} \cite{CT}\label{thm1}
Suppose the true target vector $X$ has
the sparsity at most $s$. 
Suppose the  restricted isometry constant of $\bA$
satisfies the inequality 
\beq
\label{ric}
\delta_{3s}+3\delta_{4s}< 2. 
\eeq
 Then $X$ is the unique solution of BP.
\end{theorem}
\begin{remark}
\label{rmk:sp}
Greedy algorithms have significantly lower  computational
complexity than linear programming and have 
provable performance under various conditions.
For example
under the condition  $\delta_{3s}<0.06$
 the Subspace Pursuit (SP) algorithm is guaranteed to exactly recover $X$ via a finite number of iterations \cite{DM}.
 
\end{remark}

In this appendix we show that the sensing matrix for source inversion satisfies RIP.
 This can be readily seen
by rewriting the paraxial Green function (\ref{8-3})
\beq
\label{8-3-1}
G(\br,\ba)={e^{i\om z_0}\over 4\pi z_0} 
e^{i\om (x^2+y^2)/(2z_0)} e^{-i\om x\xi /z_0}
e^{-i\om y\eta/z_0}
e^{i\om (\xi^2+\eta^2)/(2z_0)},
\eeq
for $\br=(x,y, z_0), \ba=( \xi,\eta, 0). $

Now the sensing matrix (\ref{1.10}) can be written
as the product of three matrices
\[
\bA=\bD_1 \bPhi\bD_2
\]
where 
\[
\bD_1=\hbox{diag} (e^{i\om (\xi_j^2+\eta_j^2)/(2z_0)}),\quad \bD_2=\hbox{diag} (e^{i\om (x_l^2+y^2_l)/(2z_0)})
\]
 are  unitary and  
\[
\bPhi=n^{-1/2} \lt[e^{-i\om \xi_j x_l /z_0}e^{-i\om \eta_j y_l /z_0}\rt]. 
\]

Assume without loss of generality that $x_l=y_l=l\ell, l=0,...,m-1$
and suppose that  $(\xi_j,\eta_j), j=1,...,n$ are independently and  uniformly 
distributed in  $[0, A]\times [0,A]$ with 
\beq
\label{102}
{A\ell\over \lambda z_0}=1,
\eeq
cf. (\ref{aperture}). 

The result  essential for our problem is due to Rauhut \cite{Rau}. 
 
\begin{theorem} \cite{Rau}
If 
\[
{n\over \ln{n}}\geq C s\ln^2{s} \ln{m} \ln{1\over \ep}
\]
for $\ep\in (0,1)$ and some absolute constant $C$,  then the restricted isometry condition (\ref{ric}) is satisfied 
with probability at least $1-\ep$.
\end{theorem}
See \cite{CRT1, RV} for similar results for sensors located
in a particular  discrete subset of $[0,A]\times [0,A]$.

Since $\bD_1$ and $\bD_2$ are diagonal and unitary,
$\bA$ satisfies (\ref{ric}) if and only if $\bPhi$ satisfies
the same condition. 

\begin{theorem}
\label{thm:is2}
Let the sensor array be randomly drawn from the sensor
ensemble satisfying (\ref{102}). If 
\[
{n\over \ln{n}}\geq C s\ln^2{s} \ln{m} \ln{1\over \ep}
\]
for $\ep\in (0,1)$ and some absolute constant $C$,  then 
with probability at least $1-\ep$ all source amplitudes of sparsity less than $s$
can be uniquely determined from BP.
\end{theorem}

From the relationships (\ref{50}), (\ref{51}) it follows
immediately that $\bA^{\rm SA}$ also satisfies
(\ref{ric}) if
\beq
\label{101}
{A\ell\over \lambda z_0}={1\over 2},
\eeq
cf. (\ref{49-3}). 

\begin{theorem}
\label{thm:sa2}
Let the sensor array be randomly drawn from the sensor
ensemble satisfying (\ref{101}). If 
\[
{n\over \ln{n}}\geq C s\ln^2{s} \ln{m} \ln{1\over \ep}
\]
for $\ep\in (0,1)$ and some absolute constant $C$,  then 
with probability at least $1-\ep$ all scatter amplitudes of sparsity less than $s$
can be uniquely determined from BP.
\end{theorem}

The superiority of the RIP approach, if it works, 
over that of the incoherence taken in the main text
of the paper, is that the uniqueness for BP is
guaranteed for all targets of sparsity at most $s$
and  the target ensemble needs not
be  introduced. Moreover, the stability of solution w.r.t. 
noise is guaranteed \cite{CRT2, HS09}. 
However, Theorem \ref{thm:rm} for the response matrix imaging does not seem amenable to the RIP approach.

\section{Proof of Theorem \ref{tropp}}
\label{sec:pf}
Theorem \ref{tropp} is an easy consequence from
the following two theorems due to Tropp \cite{Tropp}.

\begin{proposition}
\cite{Tropp}
\label{tropp1} Let $A$ be a $n\times m$ matrix with full rank. 
Let $\bA_s$ be a submatrix  generated by randomly
selecting $s$ columns of $\bA$. The condition
\beq
\label{sp}
6 \lt(p\mu^2s \ln{(1+s/2)}\rt)^{1/2} +{s\over m}
\|\bA\|_2^2\leq {\alpha\over 2 e^{1/4}},\quad p\geq 1
\eeq
implies that
\beq
\label{46}
\IP\lt(\|\bA_s^*\bA_s-\bI_s\|_2<\alpha\rt)
\geq 1- \lt({2\over s}\rt)^p.
\eeq
\end{proposition}

\begin{proposition}\label{tropp2} \cite{Tropp}
Let $X$ be drawn from the target ensemble.
If 
\beq
\label{43}
\mu^2 s\leq \lt(8\ln{m\over \ep}\rt)^{-1},\quad\ep\in (0,1)
\eeq
and if the least singular value
\beq
\label{44}
\sigma_{\rm min} (\bA_s)\geq 2^{-1/2},\quad |S|=s
\eeq
then $X$ is the unique solution of BP (\ref{L1}), except with probability $2\epsilon$.
\end{proposition}

First of all, (\ref{M}) and (\ref{Op}) together imply
(\ref{sp}) and (\ref{43})  with $\alpha=1/2$. Moreover,  by
Proposition \ref{tropp1}
(\ref{44}) holds with probability greater than or equal to the right hand side
of  (\ref{46}). 
Hence  we need only to derive
the claimed bound for 
the probability of the event $E$ that $X$ is the unique solution of BP.  This follows from  the estimate
\beqn
\IP(E)&\geq& \IP(E\big|\|\bA_s^*\bA_s-\bI_s\|_2< 2^{-1} )\IP(\|\bA_s^*\bA_s-\bI_s\|_2< 2^{-1})\\
&\geq &(1-2\ep)(1-(2/s)^{p})\\
&\geq& 1-2\ep-(2/s)^{p}. 
\eeqn

\section{Matched field processing}
\label{sec:A}

 Matched field processing (MFP) has been used extensively 
 for source localization in underwater acoustics and is closely
 related to the matched filter in signal processing. 
 \commentout{
 MFP can be considered an inverse problem in that localization is achieved
by comparing the received wavefront with that predicted by a model computed over the space
of hypothesized location parameters.
In principle, MFP consists of three main stages.
Stage one is to collect the measurement data
using  $n$ sensors; stage two is to select an appropriate
propagation model and use it to compute
the wave field at selected locations. The third stage of
MFP 
is to correlate the measured field with the modeled
fields. 
The propagation model is computed over a grid of rangeÐdepth cells and 
the candidate field with the highest
correlation should correspond to the true location
of the source. 
 In this paper, we assume the perfect
knowledge of the medium so that there
is no mismatch between the propagation model
and the medium. 
}

\commentout{
 The set of location parameters producing a wavefront
model which best matches the received wavefront observed by the array is the source location
estimate.
Thus to determine source location,
the model is computed over a grid of rangeÐdepth cells and each realization is compared with
the received wavefront to determine the best match.
Given  perfect
knowledge of the medium, }
The conventional MFP
uses the Bartlett processor with the
 ambiguity surface
\beq
B(\br)={\bG^*(\br)YY^*\bG(\br)\over
\|\bG(\br)\|_2^2}
\eeq
\cite{Tol}. 
The Bartlett processor is motivated by the
following optimization problem:
Maximize
the quantity
\beq
\label{snr}
W^* YY^* W
\eeq
subject to the constraint:
\[
W^* W=1.
\]
The solution 
\[
W=Y/\|Y\|_2
\]
is the weight vector for
the matched filter. 
In the case of one point source of amplitude $\sigma_1$ located at $\bx_1$, 
\[
Y=\sigma_1\bG(\br_1)
\]
hence
\beq
\label{mf1}
W={\sigma_1\bG(\br_1)\over |\sigma_1| \|\bG(\br_1)\|_2}.
\eeq
Extending (\ref{mf1}) to an arbitrary field point $\br$ 
by substituting $\br$ for  $\br_1$
we obtain  the Bartlett processor from (\ref{snr}). 

In general, $Y$ is the $n$-dimensional measurement vector consisting
the received signals of the array. 
For inverse scattering in the RM set-up, there are
$n$ measurement vectors corresponding  to
$n$  probe signals. The ambiguity surface  in this case
 is  the sum of the $n$ ambiguity surfaces for the $n$ probe
 signals. 

In contrast to the conventional matched field processor, 
the compressed sensing processor  utilizing the $\ell^1$-minimization
\cite{CRT2, Don1} or various greedy algorithms
\cite{DM, NTV, NV, Tro} are nonlinear.

\commentout{
  \section{The Foldy-Lax formulation}
  \label{sec-FL}
  \label{Appendix_FoldyLax}

Let $\bx_j,  j=1,...,s$ be the locations of
the $s$ target points.  
The Green function  $\tilde G(\mathbf{r},\ba_i)$ in the presence of  targets satisfies 
\beq
\label{fl}
\tilde G(\mathbf{r},\ba_i)=G(\br, \ba_i)+\sum_{j=1}^{s}\sigma_j G(\mathbf{r},\bx_j) \tilde G(\bx_j,\ba_i) 
\eeq
where $\sigma_j$ is the scattering amplitude of scatterer $j$. 
Evaluation at the scatterers yields
the Foldy-Lax equation: 
\beqn
\tilde G(\bx_j,\ba_i)&=&G(\bx_j,\ba_i)+\sum_{l\neq j}\sigma_l G(\bx_j,\bx_l) \tilde G(\bx_l,\ba_i),\quad j=1,...,s.
 \eeqn
Consequently, the exciting field at the target locations is determined by
\beq
\begin{pmatrix}\tilde G(\bx_1,\ba_i)\\ \tilde G(\bx_2,\ba_i)\\ \vdots \\ \tilde G(\bx_{s},\ba_i) \end{pmatrix}&=&
\mathbf{F}^{-1}
\begin{pmatrix}G(\bx_1,\ba_i)\\ G(\bx_2,\ba_i)\\ \vdots \\ G(\bx_{s},\ba_i) \end{pmatrix}\label{ill}
\eeq
where
\beqn
\mathbf{F}&=&\begin{pmatrix}1&-\sigma_2G(\bx_1,\bx_2)&\dots&-\sigma_{s}G(\bx_1,\bx_{s})\\
-\sigma_1G(\bx_2,\bx_1)&1&\dots&-\sigma_{s}G(\bx_2,\bx_{s})\\  \dots&\dots&\ddots&\dots \\
-\sigma_1G(\bx_{s},\bx_1)&-\sigma_2G(\bx_{s},\bx_2)&\dots&1\end{pmatrix}.\label{Foldy_Lax_Matrix}
\eeqn
The vector on the left hand side of (\ref{ill}) in turn determines
the scattered field through (\ref{fl}). 

Let $\bG$ be the matrix whose columns are $\bG(\bx_j), j=1,...,s$:
\[
\bG=\lt(\bG(\bx_1), \bG(\bx_2),\cdots, \bG(\bx_s)\rt)
\]
where $\bG$ is the Green vector given by (\ref{1.0})
and $\Sigma={\rm diag}(\sigma_1,...,\sigma_s)$. 
The response matrix $\bR$ is an $n\times n$ square matrix given by
\[
\bR=\bG\Sigma {\mathbf F}^{-1} \bG^t.
\]
\commentout{
\beqn
\bR= \begin{pmatrix} &&&\\ \sigma_1\bG (\bx_1)&\sigma_2\bG (\bx_2)&\dots&\sigma_{s} \bG(\bx_{s})\\&&&\end{pmatrix}
\mathbf{F}^{-1}
\begin{pmatrix} \bG^t(\bx_1)\\ \bG^t(\bx_2)\\ \vdots\\  \bG^t(\bx_{s})\end{pmatrix}
\eeqn
}

}

\end{document}